\documentclass[english]{article}
\usepackage[T1]{fontenc}
\usepackage[latin9]{inputenc}
\usepackage{amsmath,amsthm} 
\usepackage{graphicx}
\usepackage{amssymb}
\usepackage{enumitem}
\usepackage{natbib} 
\usepackage{bussproofs}
\usepackage[letterpaper]{geometry} 
\usepackage{thmtools,thm-restate}
\usepackage{tikz}
\usepackage{pgf} 
\usetikzlibrary{patterns,automata,arrows,shapes,snakes,topaths,trees,backgrounds,positioning,through,calc}
\usepackage{setspace}
\usepackage{multicol}
\usepackage{graphicx}
\usepackage{stmaryrd}  
\usepackage{hyperref}
\hypersetup{colorlinks=true, citecolor=blue, linkcolor=blue}  

\definecolor{darkgreen}{rgb}{0.0, 0.3, 0.0}
\usepackage{changes}
\colorlet{Changes@Color}{darkgreen}

\usepackage{natbib}
\setcitestyle{round}

\theoremstyle{definition}
\newtheorem{theorem}{Theorem}[section]
\newtheorem{fact}[theorem]{Fact}
\newtheorem{proposition}[theorem]{Proposition}

\newtheorem{definition}[theorem]{Definition}
\newtheorem{lemma}[theorem]{Lemma}
\newtheorem{question}[theorem]{Question}  
\newtheorem{remark}[theorem]{Remark}

\usepackage[page,toc,titletoc,title]{appendix}


\begin{document} 

\onehalfspace

   \title{An extension of May's Theorem to three alternatives: \\ axiomatizing {M}inimax voting}

 \author{Wesley H. Holliday$^\dagger$ and Eric Pacuit$^\ddagger$ \\ \\
 $\dagger$ University of California, Berkeley {\normalsize (\href{mailto:wesholliday@berkeley.edu}{wesholliday@berkeley.edu})} \\
$\ddagger$ University of Maryland {\normalsize (\href{mailto:epacuit@umd.edu}{epacuit@umd.edu})}}
 
 \date{{\small Forthcoming in \textit{Social Choice and Welfare}}}
 
\maketitle

 \begin{abstract} May's Theorem [K.~O.~May, Econometrica 20 (1952) 680-684] characterizes majority voting on two alternatives as the unique preferential voting method satisfying several simple axioms. Here we show that by adding some desirable axioms to May's axioms, we can uniquely determine how to vote on three alternatives (setting aside tiebreaking). In particular, we add two axioms stating that the voting method should mitigate \textit{spoiler effects} and   the so-called \textit{strong no show paradox}. We prove a theorem stating that any preferential voting method satisfying our enlarged set of axioms, which includes some weak homogeneity and preservation axioms, must choose from among the \textit{Minimax} winners in all three-alternative elections. When applied to more than three alternatives, our axioms also distinguish Minimax from other known voting methods that coincide with or refine Minimax for three alternatives.\end{abstract}

\section{Introduction}

 Voting on two alternatives appears unproblematic in comparison to voting on three, at least in contexts where all voters (resp.~alternatives) are to be treated equally. For two alternatives, May's Theorem (\citealt{May1952}, see Theorem \ref{May} below) and its descendants (e.g., \citealt{Tideman1986}, \citealt{Asan2002}, \citealt{Woeginger2003}, \citealt{Llamazares2006}) uniquely characterize majority voting. For three alternatives, one faces the problem of majority cycles (\citealt{Condorcet1785}) and the related Arrow Impossibility Theorem (\citealt{Arrow1963}). In an Arrovian framework in which voters submit rankings of the alternatives, on the basis of which a winning alternative is to be chosen, dozens of voting methods have been proposed for three or more alternatives.  Nonetheless, in this paper we show that by adding some desirable axioms to May's axioms, we can uniquely determine how to vote in three-alternative elections, except in rare cases involving tiebreaking.
 
 The voting method determined by our axioms (modulo tiebreaking) is the Minimax method (\citealt{Simpson1969}, \citealt{Kramer1977}), which chooses the alternative that minimizes the maximal loss, if any, in head-to-head majority comparisons with other alternatives. In two-alternative elections, Minimax reduces to majority voting.  In three-alternative elections, it apparently coincides with Condorcet's \citeyearpar{Condorcet1785} own approach to resolving cycles with three alternatives (\citealt[p.~1233]{Young1988}).  In fact, a number of sophisticated Condorcet-consistent voting methods coincide with or refine Minimax in three-alternative elections: Kemeny (\citealt{Kemeny1959}), Ranked Pairs (\citealt{Tideman1987}), Beat Path (\citealt{Schulze2011}), Split Cycle (\citealt{HP2023}), and Stable Voting (\citealt{HP2023stable}). Thus, our result provides support for this common solution to the problem of voting on three alternatives. Of course, different sets of axioms lead to characterizations of other voting methods;\footnote{See, e.g., \citealt{Young1974} and \citealt{Nitzan1981} on the Borda method, \citealt{Richelson1978}, \citealt{Ching1996}, and \citealt{Sekiguchi2012} on the Plurality method, \citealt{Merlin2003} on scoring methods, \citealt{Henriet1985} on the Copeland method,  \citealt{Freemanetal2014} on the Instant Runoff method, and \citealt{HP2021} and \citealt{Ding2024} on the Split Cycle method.} some of these characterizations may also be considered extensions of May's Theorem (e.g., \citealt{Maskin2024});  and we leave for future discussion the normative question of which set of axioms is~preferable.
 
While our main goal is to identify axioms characterizing the common solution to voting on three alternatives given by the Condorcet methods mentioned above, we will state axioms that can also be applied to any number of alternatives. Significantly, for more than three alternatives, our axioms distinguish Minimax from all other voting rules of which we are aware. Hopefully all of the Condorcet methods listed above will eventually be axiomatically characterized with no restriction on the number of alternatives, as has been done recently for Split Cycle (\citealt{Ding2024}). Our axioms may be a step in this direction for Minimax.
 
 After some preliminary setup and review of May's Theorem in Section \ref{Prelim}, we state and prove our main result in Section \ref{Main}. We then consider the independence of our axioms in Section \ref{IndependenceSection}, discuss what happens beyond three alternatives in Section \ref{BeyondThree}, and conclude in Section~\ref{ConclusionSection}.
 
\section{Preliminaries}\label{Prelim}

Our main result will be in a setting that allows the set of voters and the set of alternatives to vary from election to election, and May's Theorem can also be stated in this setting. 

Fix sets $\mathcal{V}$ and $\mathcal{X}$ of \textit{voters} and \textit{alternatives}, respectively. Assume $\mathcal{V}$ is countably infinite and $|\mathcal{X}|\geq 3$. Given nonempty finite subsets $V\subseteq\mathcal{V}$ and $X\subseteq\mathcal{X}$, a \textit{$(V,X)$-profile} is a function $\mathbf{P}$ assigning to each $i\in V$ a strict weak order\footnote{A \textit{strict weak order} on a set $X$ is a binary relation $P$ that is asymmetric (if $(a,b)\in P$, then $(b,a)\not\in P$) and negatively transitive (if $(a,b)\not\in P$ and $(b,c)\not\in P$, then $(a,c)\not\in P$). It follows that $P$ is transitive. A strict weak order is simply the strict part of a complete and transitive relation, i.e., a ranking allowing ties. However, we will not actually need the binary relations $\mathbf{P}(i)$ to satisfy negative transitivity or even transitivity for our results to hold, as explained in Remark \ref{Domains}.} $\mathbf{P}(i)$ on $X$. If $(a,b)\in \mathbf{P}(i)$, we say that $i$ \textit{ranks $a$ above} $b$; if $(a,b)\not\in \mathbf{P}(i)$ and $(b,a)\not\in \mathbf{P}(i)$, we say that $i$ \textit{ranks $a$ and $b$ as tied}. We write $V(\mathbf{P})$ for the relevant set $V$ of voters and $X(\mathbf{P})$ for the relevant set $X$ of alternatives. A \textit{profile} is a $(V,X)$-profile for some nonempty finite $V\subseteq\mathcal{V}$ and $X\subseteq\mathcal{X}$. Given a profile $\mathbf{P}$ and set $Y\subsetneq X(\mathbf{P})$ of alternatives, let $\mathbf{P}_{\mid Y}$ be the profile obtained by restricting each $\mathbf{P}(i)$ to $Y$. Given $b\in X(\mathbf{P})$, let $\mathbf{P}_{-b}=\mathbf{P}_{\mid X(\mathbf{P})\setminus\{b\}}$.

A \textit{voting method} is a function $F$ assigning to each profile $\mathbf{P}$ in some set $\mathrm{dom}(F)$ of profiles a nonempty subset $F(\mathbf{P})\subseteq X(\mathbf{P})$. Given voting methods $F$ and $G$, we say that $F$ \textit{refines} $G$ on some domain $\mathcal{D}$ of profiles if for every $\mathbf{P}\in\mathcal{D}$, we have $F(\mathbf{P})\subseteq G(\mathbf{P})$. Note that every voting method refines itself.

According to the method $Maj$ of \textit{majority voting} on the domain of two-alternative profiles, for any such  profile $\mathbf{P}$ and alternatives $a,b\in X(\mathbf{P})$, we have $a\in Maj(\mathbf{P})$ if and only if in $\mathbf{P}$ at least as many voters rank $a$  above $b$ as rank $b$ above $a$. To state May's Theorem characterizing this voting method, we need the three axioms on voting methods in Definition~\ref{MayAx}. Informally, anonymity says that the names of voters do not matter (so all voters are treated equally); neutrality says that the names of alternatives do not matter (so all alternatives are treated equally); and positive responsiveness says that improving a winning alternative's place in a voter's ranking not only keeps that alternative a winner\footnote{As required by so-called \textit{non-negative responsiveness} (see, e.g., \citealt{Tideman1987}).} but also makes that alternative the unique winner (so ties are easily broken).

\begin{definition}\label{MayAx} Let $F$ be a voting method.
\begin{enumerate}
\item $F$ satisfies \textit{anonymity} if for any $\mathbf{P},\mathbf{P}'\in\mathrm{dom}(F)$ with $X(\mathbf{P})=X(\mathbf{P}')$, if there is a permutation $\pi$ of $\mathcal{V}$ such that $\pi[V(\mathbf{P})]=V(\mathbf{P}')$ and $\mathbf{P}(i)=\mathbf{P}'(\pi(i))$ for all $i\in V(\mathbf{P})$, then $F(\mathbf{P})=F(\mathbf{P}')$.
 \item $F$ satisfies \textit{neutrality} if for any $\mathbf{P},\mathbf{P}'\in\mathrm{dom}(F)$ with $V(\mathbf{P})=V(\mathbf{P}')$, if there is a permutation $\tau$ of $\mathcal{X}$ such that $\tau[X(\mathbf{P})]=X(\mathbf{P}')$ and for each $i\in V(\mathbf{P})$ and $a,b\in X(\mathbf{P})$, $(a,b)\in \mathbf{P}(i)$ if and only if $(\tau(a),\tau(b))\in \mathbf{P}'(i)$, then $\tau[F(\mathbf{P})]= F(\mathbf{P}')$.
\item $F$ satisfies \textit{weak positive responsiveness} if for any  $\mathbf{P},\mathbf{P}'\in\mathrm{dom}(F)$ and $a\in X(\mathbf{P})$, if $a\in F(\mathbf{P})$ and $\mathbf{P}'$ is obtained from $\mathbf{P}$ by one voter in $V(\mathbf{P})$ who ranked $a$ uniquely last in $\mathbf{P}$ switching to ranking $a$ uniquely first in $\mathbf{P}'$, otherwise keeping their ranking the same, then $F(\mathbf{P}')=\{a\}$.
\end{enumerate}
\end{definition}

\begin{remark}\label{PRremark} We add the adjective `weak' in `weak positive responsiveness' because of our assumption that the relevant voter ranked alternative $a$ uniquely \textit{last} and switches to ranking $a$ uniquely \textit{first}. This makes our statement weaker than May's \citeyearpar{May1952} original statement of positive responsiveness in two-alternative elections, which implies that if $a\in F(\mathbf{P})$ and the voter ranked $a$ and $b$ as tied, then their moving $a$ above $b$ ensures $F(\mathbf{P}')=\{a\}$, and that if the voter ranked $b$ above $a$, then their making $a$ and $b$ tied ensures $F(\mathbf{P}')=\{a\}$. Our statement is also weaker than the generalization of positive responsiveness to elections with more than two alternatives in \citealt[Def.~7]{Barbera1977}, even if voters' rankings do not contain ties, since that generalization states that if $a\in F(\mathbf{P})$, then a voter moving $a$ up at all in their ranking ensures $F(\mathbf{P}')=\{a\}$, whereas we require that the voter moves $a$ from uniquely last to uniquely first, which will matter when we turn to three-alternative elections in Section~\ref{Main}.\end{remark}

We are now ready to state May's Theorem. We include a proof sketch in order to keep the paper self-contained.

\begin{theorem}[\citealt{May1952}]\label{May} Let $F$ be a voting method on the domain of two-alternative profiles. Then the following are equivalent:
\begin{enumerate}
\item\label{May1} $F$ satisfies anonymity, neutrality, and weak positive responsiveness;
\item\label{May2} $F$ is majority voting.
\end{enumerate}
\end{theorem}
\begin{proof} (Sketch) The implication from \ref{May2} to \ref{May1} is obvious. For the implication from \ref{May1} to \ref{May2}, we consider two cases. Case 1: in $\mathbf{P}$, the same number of voters rank $a$ above $b$ as rank $b$ above $a$. Then anonymity and neutrality imply $F(\mathbf{P})=\{a,b\}$, in agreement with majority voting.

Case 2: in $\mathbf{P}$, $n$ voters rank $b$ above $a$, $n+k$ (for $k>0$) rank $a$ above $b$, and the remaining $\ell$ voters rank them as tied. Suppose for contradiction that $b\in F(\mathbf{P})$. Then by anonymity and neutrality, in any profile $\mathbf{Q}$ in which $n$ voters rank $a$ above $b$, $n+k$ rank $b$ above $a$, and the remaining $\ell$ voters rank them as tied,  we have $a\in F(\mathbf{Q})$. Now let $\mathbf{P}'$ be obtained from $\mathbf{P}$ by flipping $k$ voters who rank $a$ above $b$ to rank $b$ above $a$. Then $\mathbf{P}'$ is a $\mathbf{Q}$ as above, so $a\in F(\mathbf{P}')$. But by repeated application of weak positive responsiveness starting from $\mathbf{P}$, $F(\mathbf{P}')=\{b\}$, so we have a contradiction. Hence $b\not\in F(\mathbf{P})$, so $F(\mathbf{P})=\{a\}$ in agreement with majority~voting.\end{proof}

\section{Main result}\label{Main}

Given a profile $\mathbf{P}$ and $a,b\in X(\mathbf{P})$, the \textit{margin of $a$ over $b$ in $\mathbf{P}$},  $Margin_\mathbf{P}(a,b)$, is the number of voters who rank $a$ above~$b$ in $\mathbf{P}$ minus the number of voters who rank $b$ above $a$ in $\mathbf{P}$. According to the Minimax voting method, in a profile $\mathbf{P}$, each alternative $a$ receives a score equal to the maximum margin that any alternative $b$ has over $a$, so lower scores are better; then $Minimax(\mathbf{P})$ is the set of alternatives with minimal score. For example, in a profile with three alternatives, $a$, $b$, and $c$, if the margin of $a$ over $b$ is 10, the margin of $b$ over $c$ is 6, and the margin of $c$ over $a$ is 8, then $Minimax(\mathbf{P})=\{c\}$. In a profile in which one alternative has positive margins over all others, this alternative---called the \textit{Condorcet winner}---is the unique Minimax winner. In fact, it is easy to see that in arbitrary profiles, a Minimax winner can be characterized as an alternative $m$  that is ``closest'' to being a Condorcet winner in the sense that $m$ can be made the Condorcet winner with the addition of the fewest number of voters to the profile (cf.~\citealt[\S~5]{Young1977}). If there are multiple alternatives with the same Minimax score, then there are multiple Minimax winners; but as the number of voters goes to infinity, the proportion of profiles in which there are multiple Minimax winners goes to zero (see, e.g., \citealt{HT2020}).

\begin{fact}\label{MinimaxMay} Minimax satisfies anonymity, neutrality, and weak positive responsiveness.
\end{fact}
\begin{proof} Anonymity and neutrality are immediate from the definition of Minimax. For weak positive responsiveness, if $a\in Minimax(\mathbf{P})$, so $a$ has a minimal Minimax score in $\mathbf{P}$, and $\mathbf{P}'$~is~obtained from $\mathbf{P}$ by one voter in $V(\mathbf{P})$ who ranked $a$ uniquely last in $\mathbf{P}$ switching to ranking $a$ uniquely first in $\mathbf{P}'$, then $a$'s score decreases from $\mathbf{P}$ to $\mathbf{P}'$, while other alternatives' scores do not decrease, so $a$ has the unique minimum score in $\mathbf{P}'$, so $Minimax(\mathbf{P}')=\{a\}$.
\end{proof}
\noindent In the Appendix, we also consider a refinement of Minimax that satisfies full positive responsiveness as in Remark \ref{PRremark}, but in our main theorem (Theorem \ref{MainThm}) it suffices to assume just weak positive responsiveness.

In the rest of this section, we extend May's Theorem to the domain of \textit{three-alternative} profiles: we add axioms alongside May's  in order to force a voting method $F$ to refine Minimax, i.e., to choose from among the Minimax winners, on this domain. Thus, whenever Minimax selects a unique winner, $F$ will select the same unique winner; in the unlikely (with many voters) event of a Minimax tie, $F$ may perform further tie-breaking. Before adding axioms alongside May's, we stress the importance of weak positive responsiveness in what follows: \textit{all the other axioms in our main theorem are satisfied by the trivial voting method that selects all alternatives in every profile}, as well as by other methods that are non-trivial but still produce many ties, e.g., the Weighted Covering method of \citealt{Dutta1999} and  \citealt{Fernandez2018} (see \citealt[Proposition~3.16]{Ding2024}). Thus, weak positive responsiveness is what will give our other axioms their ``bite.''

The key additional axioms are the following, which are expressions of the idea that a voting method should respond well to the addition of \textit{new voters} and the addition of \textit{new alternatives} to an election.

\begin{definition}\label{AddAxioms} Let $F$ be a voting method.
\begin{enumerate}
\item\label{AddAxioms1} $F$ satisfies \textit{positive involvement} if for any  $\mathbf{P},\mathbf{P}'\in\mathrm{dom}(F)$ and $a\in X(\mathbf{P})$, if $a\in F(\mathbf{P})$ and $\mathbf{P}'$ is obtained from $\mathbf{P}$ by adding one new voter who ranks $a$ uniquely first, then $a\in F(\mathbf{P}')$.
\item\label{AddAxioms2} $F$ satisfies \textit{immunity to spoilers} if for any $\mathbf{P}\in\mathrm{dom}(F)$ and $a,b\in X(\mathbf{P})$, if $a\in F(\mathbf{P}_{-b})$, $F(\mathbf{P}_{\mid \{a,b\}})=\{a\}$, and $b\not\in F(\mathbf{P})$, then $a\in F(\mathbf{P})$. 
\end{enumerate}
The axiom of \textit{near immunity to spoilers} is defined as in \ref{AddAxioms2} but with $a\in F(\mathbf{P}_{-b})$ replaced by $F(\mathbf{P}_{-b})=\{a\}$.
\end{definition}

Positive involvement (\citealt{Saari1995}) is a variable-voter axiom: adding to an election a voter who ranks $a$ uniquely first should not cause $a$ to go from winning to losing. P\'{e}rez \citeyearpar{Perez2001} dubs violations of positive involvement the  ``(Positive) Strong No Show Paradox.'' Surprisingly, many well-known voting methods violate positive involvement (see \citealt{Perez2001}  and \citealt{HP2021b}), which P\'{e}rez calls ``a common flaw in Condorcet voting correspondences.'' Recall that a voting method satisfies \textit{Condorcet consistency} (or the \textit{Condorcet winner criterion}) if it always picks the Condorcet winner uniquely when one exists. Well-known Condorcet-consistent methods that violate positive involvement already for three alternatives include Black's \citeyearpar{Black1958} method (if there is no Condorcet winner, pick the Borda winner (\citealt{Borda1781})), Condorcet-Plurality (if there is no Condorcet winner, pick the alternative with the most first-place votes; see Lemma~\ref{PIind}), Condorcet-Hare (if there is no Condorcet winner, pick the winner according to the Hare method (\citealt{Hare1859}), also known as Instant Runoff Voting), and more (e.g., Baldwin (\citealt{Baldwin1926}), the version of Nanson (\citealt{Nanson1882}) that iteratively eliminates alternatives with below average Borda score (\citealt[p.~37]{Zwicker2016}), Raynaud (\citealt{Arrow1986}), Condorcet-Bucklin, Condorcet-Coombs, etc.).\footnote{For examples of violations, see \href{https://github.com/wesholliday/may-minimax}{https://github.com/wesholliday/may-minimax}.} By contrast, though Minimax is Condorcet consistent, it does not have the flaw of violating positive involvement.\footnote{P\'{e}rez  \citeyearpar[Lemma 3]{Perez1995} (cf.~\citealt[Lemma 1]{Perez2001}), adapting an argument of Moulin \citeyearpar{Moulin1988}, also shows that if a voting method $F$ satisfies positive involvement and Condorcet consistency, then in any  profile $\mathbf{P}$ of linear ballots (without ties), $F$ must choose from among the alternatives $a$ that Holliday \citeyearpar{Holliday2024} calls \textit{defensible}: for any alternative $b$, there is an alternative $c$ such that $\mathrm{Margin}_\mathbf{P}(c,b)\geq \mathrm{Margin}_\mathbf{P}(b,a)$ (see \citealt{Kasper2019} for a study of the voting method that selects all defensible alternatives). For profiles of strict weak order ballots (allowing ties), there is an analogous result if Condorcet consistency is strengthened to the condition that if there is a weak Condorcet winner (an alternative with non-negative margins against every alternative), then $F$ picks from among the weak Condorcet winners (\citealt[Lemma 2.2]{Holliday2024}). However, there can be Minimax losers who are defensible (see the ``descending case'' in the proof of Theorem \ref{MainThm}), so weak Condorcet consistency and positive involvement are not sufficient for our purposes here. We will use weak positive responsiveness and two axioms to come (block preservation and homogeneity) to show that defensible alternatives that are Minimax losers will not be selected.}

\begin{fact}[\citealt{Perez2001}]\label{MinimaxPI} Minimax satisfies positive involvement.
\end{fact}
\begin{proof} Observe that $a$'s Minimax score decreases by $1$ with the addition of a new voter who ranks $a$ uniquely first, while all other Minimax scores decrease by at most $1$ with such an addition, so if $a$ had minimal score before the addition of the new voter, then $a$ still has minimal score after the addition of the new voter.
\end{proof}

Immunity to spoilers (\citealt{HP2023}) is a variable-alternative axiom: adding to an election a losing alternative $b$ should not cause the loss of an alternative $a$ that would win without $b$ in the election and would win in a head-to-head contest with $b$.\footnote{In \citealt{HP2023}, the condition in immunity to spoilers that $a$ would win in a head-to-head contest with $b$ is expressed as the condition that more voters rank $a$ above $b$ than vice versa;  but in the presence of May's axioms that imply majority rule in two-alternative elections, the two formulations of immunity to spoilers are equivalent.} The \textit{near immunity} version in Definition \ref{AddAxioms} adds the assumption that $a$ is the \textit{unique} winner without $b$ in the election, which makes for a weaker axiom.\footnote{Near immunity may be more plausible in cases where some distinct alternatives $a$ and $c$ have a zero margin. Then if both $a$ and $c$ beat a third alternative $b$, but $c$ beats $b$ by a larger margin, then it may be reasonable to break the tie between $a$ and $c$ in favor of $c$. This is consistent with near immunity to spoilers but not with the stronger version. For further discussion, see the last paragraph of Section 1.1 and Definition 4.8.1 in  \citealt{HP2023}, as well as the Appendix of this paper.} Many well-known voting methods violate (near) immunity to spoilers (see \citealt{HP2023}). Although ranked ballots are not collected in Plurality Voting elections, the 2000 Florida Presidential Election is plausibly a case of a spoiler effect for Plurality relative to voters' actual preferences (\citealt{Magee2003}, \citealt{Herron2007}).  Spoiler effects for Instant Runoff Voting include the 2009 Mayoral Election in Burlington, Vermont and the 2022 Alaska Primary Special Election for U.S. Representative (see \href{https://github.com/voting-tools/election-analysis}{https://github.com/voting-tools/election-analysis}). By contrast, Minimax renders such spoiler effects impossible.

\begin{fact}\label{MinimaxIS} Minimax satisfies immunity to spoilers.
\end{fact}
\begin{proof} Suppose $a$ wins in $\mathbf{P}_{-b}$, so $a$'s Minimax score is minimal in $\mathbf{P}_{-b}$.  If $a$ is the unique winner in $\mathbf{P}_{\mid \{a,b\}}$, so more voters rank $a$ above $b$ in $\mathbf{P}$ than vice versa, then $a$'s Minimax score does not increase from $\mathbf{P}_{-b}$ to $\mathbf{P}$. Moreover, no other alternative's Minimax score can decrease from $\mathbf{P}_{-b}$ to $\mathbf{P}$ by the definition of the Minimax score. Thus, $a$'s Minimax score remains minimal in $\mathbf{P}$ unless $b$ has lower Minimax score than $a$ in $\mathbf{P}$. But by assumption, $b\not\in Minimax(\mathbf{P})$. Hence $a$'s Minimax score remains minimal in $\mathbf{P}$, so $a\in Minimax(\mathbf{P})$.
\end{proof}

\begin{remark} It is important to note that the use of near immunity to spoilers in our main result (Theorem~\ref{MainThm}) could be replaced by Condorcet consistency (Proposition \ref{CondorcetVersion}). We use the framing in terms of spoilers for three reasons. First, near immunity to spoilers is strictly weaker than Condorcet consistency for three alternatives (while Condorcet consistency implies near immunity to spoilers on this domain, the converse implication fails, as shown by, e.g., the trivial voting method that selects all alternatives in every profile or by the Weighted Covering method mentioned above). Second, the framing in terms of immunity to spoilers highlights the desirability of using Minimax for elections with three alternatives, such as the top-three elections advocated in \citealt{Foley2024}, since mitigating spoiler effects is a significant benefit. Third, looking ahead to profiles with more than three alternatives (Section \ref{BeyondThree}), Condorcet consistency does not imply near immunity to spoilers for more than three alternatives, and  Minimax is one of the rare voting methods (along with Split Cycle) satisfying the axiom for any number of alternatives (see Table \ref{AxFig}).\end{remark}

In addition to the substantive axioms of positive involvement and immunity to spoilers, we will use two relatively uncontroversial axioms for our proof. Where $\mathbf{P}$ and $\mathbf{P}'$ are profiles with $V(\mathbf{P})\cap V(\mathbf{P}')=\varnothing$, $\mathbf{P}+\mathbf{P}'$ is the combined profile with $V(\mathbf{P}+\mathbf{P}')=V(\mathbf{P})\cup V(\mathbf{P}')$. Let $2\mathbf{P}$ be $\mathbf{P}+\mathbf{P}'$ where $\mathbf{P}'$ is a copy of $\mathbf{P}$ with a disjoint set of voters.\footnote{Assume a well-ordering $\leqslant$ of $\mathcal{V}$ and pick the first $|V(\mathbf{P})|$-many voters according to $\leqslant$ that do not belong to $V(\mathbf{P})$. For $n>2$, let $n\mathbf{P}=(n-1)\mathbf{P}+\mathbf{P}'$ where $\mathbf{P}'$ is a copy of $\mathbf{P}$ with a set of voters disjoint from that of $(n-1)\mathbf{P}$.}

\begin{definition}\label{HomPres} Let $F$ be a voting method. 
\begin{enumerate}
\item\label{Hom} $F$ satisfies (\textit{inclusion}) \textit{homogeneity} if for any $\mathbf{P}\in\mathrm{dom}(F)$, we have $2\mathbf{P}\in \mathrm{dom}(F)$ and $F(\mathbf{P})\subseteq F(2\mathbf{P})$.
\item $F$ satisfies \textit{block preservation} if for any $\mathbf{P},\mathbf{P}'\in \mathrm{dom}(F)$, if $X(\mathbf{P})=X(\mathbf{P}')$, $V(\mathbf{P})\cap V(\mathbf{P}')=\varnothing$, and $\mathbf{P}'$ contains for each linear order on $X(\mathbf{P})$ exactly one voter submitting that linear order and no other voters, then  $\mathbf{P}+\mathbf{P}'\in \mathrm{dom}(F)$ and $F(\mathbf{P})\subseteq F(\mathbf{P}+\mathbf{P}')$.
\end{enumerate}
\end{definition}
\noindent The homogeneity axiom of \citealt{Smith1973} says that for any positive integer $n$, replacing each voter with $n$ copies of that voter with the same ranking does not change the election outcome. The weaker version in Definition~\ref{HomPres}.\ref{Hom} only concerns doubling and only requires that the winners in $\mathbf{P}$ are included among the winners in $2\mathbf{P}$ (so, e.g., it allows a voting method to select all alternatives in profiles with an even number of voters and do something else in profiles with an odd number of voters). Minimax satisfies Smith's stronger version, but the weaker axiom  suffices for our axiomatic characterization result. Readers who regard this distinction as splitting hairs may henceforth simply think in terms of Smith's  version of homogeneity.

Block preservation says that adding a block of voters who submit one copy of each linear order does not cause any alternative that would have won before to lose after. For three-alternative profiles, this simply means that adding six voters with the rankings $abc$, $acb$, $bac$, $bca$, $cab$, and $cba$ does not cause any previous winner to lose. Minimax satisfies a stronger \textit{block invariance} axiom according to which the set of winners cannot change at all, but again it suffices for our main result to assume only the weaker axiom.\footnote{We do not know of other papers discussing block preservation or block invariance, but they can be derived from other well-known axioms, such as the conjunction of anonymity, neutrality, and \textit{reinforcement}: if  $X(\mathbf{P})=X(\mathbf{P}')$ and $V(\mathbf{P})\cap V(\mathbf{P}')=\varnothing$, then $F(\mathbf{P})\cap F(\mathbf{P}')\neq\varnothing$ implies $F(\mathbf{P}+\mathbf{P}')=F(\mathbf{P})\cap F(\mathbf{P}')$. Anonymity and neutrality  together imply that if $\mathbf{P}'$ consists of exactly one copy of each linear order on $X(\mathbf{P})$, then $F(\mathbf{P}')=X(\mathbf{P}')$ (as does the \textit{cancellation} property of \citealt{Young1974}), in which case reinforcement yields $F(\mathbf{P}+\mathbf{P}')= F(\mathbf{P})\cap X(\mathbf{P}')=F(\mathbf{P})$. However, no voting method that always picks the Condorcet winner uniquely when one exists can satisfy reinforcement (see \citealt[Prop.~2.5]{Zwicker2016}, and cf.~Remark 4.26 of \citealt{HP2023} for doubts about the plausibility of reinforcement), whereas block invariance is widely satisfied.}

\begin{fact}\label{MinimaxHBP} Minimax satisfies homogeneity and block preservation.
\end{fact}
\begin{proof} The transformation from $\mathbf{P}$ to $2\mathbf{P}$ for homogeneity simply doubles all Minimax scores, so the alternatives with minimal score do not change. The transformation from $\mathbf{P}$ to $\mathbf{P}+\mathbf{P}'$ for block preservation does not change scores at all. \end{proof}

As far as we know, the only voting method in the literature that violates homogeneity is the Dodgson method (see \citealt{Brandt2009}).\footnote{Brandt \citeyearpar{Brandt2009} gives a profile $\mathbf{P}$ for which $Dodgson(\mathbf{P})\neq Dodgson(3\mathbf{P})$; observing in addition that $Dodgson(2\mathbf{P})\not\subseteq Dodgson(4\mathbf{P})$ yields a violation of (inclusion) homogeneity in our formulation.} Violations of block preservation are slightly more common, coming not only from Dodgson but also from Bucklin (see \citealt{Hoag1926}), Coombs (\citealt{Coombs1964}, \citealt{Grofman2004}), and Mediancentre-Borda (\citealt{Cervone2012}).\footnote{For Bucklin, let $\mathbf{P}$ contain $1$ voter with the ranking $abc$ and $3$ with  $cab$. Then since $c$ has a majority of first-place votes, $c$ is the Bucklin winner. However, where $\mathbf{P}'$ consists of each of the six linear orders over $\{a,b,c\}$, in $\mathbf{P}+\mathbf{P}'$ no alternative receives a majority of first-place votes, but $a$ and $c$ both receive a majority of first-or-second-place votes, and $a$ has more first-or-second-place votes ($8$) than $c$ does ($7$), so $a$ is the Bucklin winner. For Coombs, let $\mathbf{P}$ contain $2$ voters with the ranking $abc$, $2$ with $cba$, $2$ with $cab$, and $1$ with $bac$. Then $c$ has a majority of first-place votes, so $c$ is the Coombs winner. But where $\mathbf{P}'$ consists of each of the six linear orders over $\{a,b,c\}$, in $\mathbf{P}+\mathbf{P}'$ no alternative receives a majority of first-place votes, and $c$ has the most last-place votes, so $c$ is eliminated in the first round (and then $a$ wins). Finally, for Mediancentre-Borda, start with the profile labeled ``Irresolute OWM Failure'' in the simulation accompanying \citealt{Cervone2012} at   \href{https://omega.math.union.edu/research/2010-05-voting/}{https://omega.math.union.edu/research/2010-05-voting/}; then adding one copy of each linear order causes the original winner to lose.} However, clearly all  scoring rules (\citealt{Young1975}) and voting methods based on majority margins (see, e.g., \citealt{HP2023}) satisfy block preservation. Moreover, the methods cited above that violate block preservation also violate immunity to spoilers, so block preservation is not doing any special work by ruling out those methods.

We are now ready to state and prove our main result.\footnote{Just before submitting the final version of this paper to the publisher, we learned from Chris Dong (personal communication) that our preprint of this paper (\citealt{HPMay}) inspired a related theorem in his recent co-authored preprint, \citealt{Brandt2024}. The theorem in \citealt{Brandt2024} (Theorem~5) replaces May's axioms and block preservation in Theorem \ref{MainThm} by an axiom of \textit{singleton negative involvement} (see their Proposition 1): if $F(\mathbf{P})\neq \{a\}$, and $\mathbf{P}'$ is obtained from $\mathbf{P}$ by adding one voter who ranks $a$ uniquely last, then $F(\mathbf{P}')\neq \{a\}$. Their theorem also replaces near immunity to spoilers by Condorcet consistency, but note as in Proposition \ref{CondorcetVersion} below that Theorem \ref{MainThm} can be formulated in this way as well.}

\begin{theorem}\label{MainThm} Let $F$ be a voting method on the domain of profiles with two or three alternatives such that 
\begin{itemize}
\item $F$ satisfies anonymity, neutrality, and weak positive responsiveness. 
\end{itemize}
If in addition 
\begin{itemize}
\item $F$ satisfies positive involvement, near immunity to spoilers,  homogeneity, and  block preservation,
\end{itemize} then $F$ refines Minimax.
\end{theorem}

\begin{proof} Assume $F$ satisfies the listed axioms. Then by Theorem \ref{May}, $F$ agrees with majority voting and hence with Minimax on all two-alternative profiles. Next we claim that for any three-alternative profile $\mathbf{P}$ in which a Condorcet winner $c$ exists, we have $F(\mathbf{P})=\{c\}$, in agreement with Minimax. Suppose $a\in F(\mathbf{P})$ for $a\neq c$. Then by homogeneity twice, $a\in F(4\mathbf{P})$. Let $\mathbf{P}'$ be obtained from $4\mathbf{P}$ by adding a block of all linear orders, so by block preservation, $a\in F(\mathbf{P}')$. Let $\mathbf{P}''$ be obtained from $\mathbf{P}'$ by changing one $cba$ ranking to $acb$, so by weak positive responsiveness, $F(\mathbf{P}'')=\{a\}$. Toward applying near immunity to spoilers, observe that since the margin of $c$ over $a$ in $\mathbf{P}$ is at least $1$, the margin of $c$ over $a$ in $4\mathbf{P}$ and hence $\mathbf{P}'$ is at least 4, so the margin of $c$ over $a$ in $\mathbf{P}''$ is at least $2$. Hence by Theorem~\ref{May}, we have $F(\mathbf{P}''_{-b})=\{c\}$. In addition,  since $c$ has a positive margin over $b$ in $\mathbf{P}''$, we have $F(\mathbf{P}''_{\mid \{c,b\}})=\{c\}$ by Theorem~\ref{May}. Then since  $b\not\in F(\mathbf{P}'')$ (given $F(\mathbf{P}'')=\{a\}$), it follows by near immunity to spoilers that $c\in F(\mathbf{P}'')$, contradicting $F(\mathbf{P}'')=\{a\}$.

Now suppose $\mathbf{P}$ is a three-alternative profile with no Condorcet winner. The \textit{margin graph} of a profile $\mathbf{P}$ is the weighted directed graph with an edge from $a$ to $b$ if $a$ has a positive margin over $b$, weighted by that margin. It follows that the margin graph of $\mathbf{P}$ is of one of the following four forms, where the dashed arrows might not be present, due to a possibly zero margin, and the absence of an arrow indicates a zero margin:
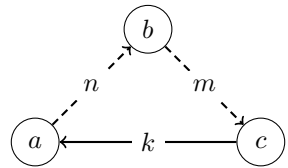
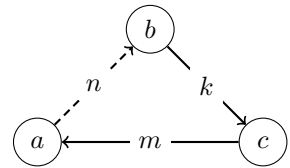
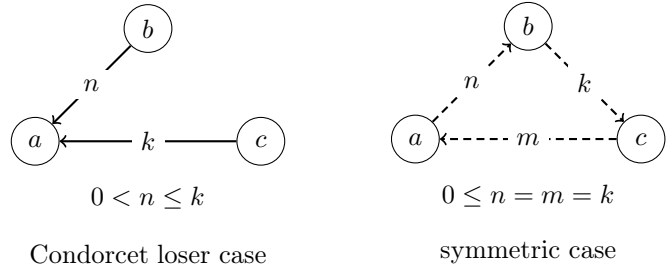
\begin{center}
\begin{tikzpicture}

\node[circle,draw, minimum width=0.25in] at (0,0) (a) {$a$}; 
\node[circle,draw,minimum width=0.25in] at (3,0) (c) {$c$}; 
\node[circle,draw,minimum width=0.25in] at (1.5,1.5) (b) {$b$}; 

\path[->,draw,thick, densely dashed] (b) to node[fill=white] {$m$} (c);
\path[->,draw,thick] (c) to node[fill=white] {$k$} (a);
\path[->,draw,thick, densely dashed] (a) to node[fill=white] {$n$} (b);

\node at (1.5,-.75) {{$0\leq n\leq m <k$}};
\node at (1.5,-1.5) {{ascending case}};

\end{tikzpicture}\qquad\qquad\begin{tikzpicture}

\node[circle,draw, minimum width=0.25in] at (0,0) (a) {$a$}; 
\node[circle,draw,minimum width=0.25in] at (3,0) (c) {$c$}; 
\node[circle,draw,minimum width=0.25in] at (1.5,1.5) (b) {$b$}; 

\path[->,draw,thick] (b) to node[fill=white] {$k$} (c);
\path[->,draw,thick] (c) to node[fill=white] {$m$} (a);
\path[->,draw,thick,densely dashed] (a) to node[fill=white] {$n$} (b);

\node at (1.5,-.75) {{$0\leq n<m\leq k$}};
\node at (1.5,-1.5) {{descending case}};

\end{tikzpicture}\vspace{.2in}

\begin{tikzpicture}

\node[circle,draw, minimum width=0.25in] at (0,0) (a) {$a$}; 
\node[circle,draw,minimum width=0.25in] at (3,0) (c) {$c$}; 
\node[circle,draw,minimum width=0.25in] at (1.5,1.5) (b) {$b$}; 

\path[->,draw,thick] (b) to node[fill=white] {$n$} (a);
\path[->,draw,thick] (c) to node[fill=white] {$k$} (a);

\node at (1.5,-.75) {{$0< n\leq k$}};
\node at (1.5,-1.5) {{Condorcet loser case}};

\end{tikzpicture}\qquad\qquad\begin{tikzpicture}

\node[circle,draw, minimum width=0.25in] at (0,0) (a) {$a$}; 
\node[circle,draw,minimum width=0.25in] at (3,0) (c) {$c$}; 
\node[circle,draw,minimum width=0.25in] at (1.5,1.5) (b) {$b$}; 

\path[->,draw,thick, densely dashed] (b) to node[fill=white] {$k$} (c);
\path[->,draw,thick, densely dashed] (c) to node[fill=white] {$m$} (a);
\path[->,draw,thick, densely dashed] (a) to node[fill=white] {$n$} (b);

\node at (1.5,-.75) {{$0\leq n=m=k$}};
\node at (1.5,-1.5) {{symmetric case}};

\end{tikzpicture}
\end{center}
Note that we use `$a$', `$b$', and `$c$' above as variables---so in the ascending case, we simply call the alternative with the largest loss `$a$', etc. To see that the above four cases are exhaustive, assuming there is no Condorcet winner, observe that if the number of edges in the margin graph is zero, then we are in symmetric case; if the number of edges is one, then we are in the ascending case with $0=n=m$; if the number of edges is two, then we are in either the Condorcet loser case or the ascending case or descending case with $0=n<m$; and if the number of edges is three, then we are in either the ascending, descending, or symmetric case.

Case 1: the margin graph of $\mathbf{P}$ is as in the \textit{symmetric case}. Then $\mathrm{Minimax}(\mathbf{P})=\{a,b,c\}\supseteq F(\mathbf{P})$.

Case 2: the margin graph of $\mathbf{P}$ is as in the \textit{ascending case}. We claim that $F(\mathbf{P})\subseteq\{b,c\}$.  Suppose for contradiction that $a\in F(\mathbf{P})$. By homogeneity, $a\in F(2\mathbf{P})$. Let $\mathbf{P}'$ be obtained from $2\mathbf{P}$ by adding $2m+1$ voters with the ranking $acb$. Since $m<k$, $2k- 2m-1>0$, so the margin graph of $\mathbf{P}'$ is 
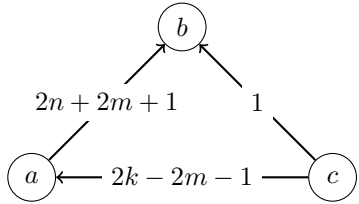
\begin{center}
\begin{tikzpicture}

\node[circle,draw, minimum width=0.25in] at (0,0) (a) {$a$}; 
\node[circle,draw,minimum width=0.25in] at (4,0) (c) {$c$}; 
\node[circle,draw,minimum width=0.25in] at (2,2) (b) {$b$}; 

\path[->,draw,thick] (c) to node[fill=white] {$1$} (b);
\path[->,draw,thick] (c) to node[fill=white] {$2k-2m-1$} (a);
\path[->,draw,thick] (a) to node[fill=white] {$2n+2m+1$} (b);

\end{tikzpicture}
\end{center}
By repeated application of positive involvement, $a\in F(\mathbf{P}')$. Yet $c$ is a Condorcet winner, so by the first paragraph of the proof, $F(\mathbf{P}')=\{c\}$, so we have a contradiction. 

Thus, $F(\mathbf{P})\subseteq\{b,c\}$. Now if $n=m$, then $\mathrm{Minimax}(\mathbf{P})=\{b,c\}$, so $F(\mathbf{P})\subseteq \mathrm{Minimax}(\mathbf{P})$, so we are done.

Now suppose $n<m$, in which case we call this a \textit{strictly ascending case}. Then we claim that $F(\mathbf{P})=\{b\}$ in line with Minimax. Suppose for contradiction that $c\in F(\mathbf{P})$. By homogeneity, $c\in F(2\mathbf{P})$.  The margin graph of $2\mathbf{P}$~is
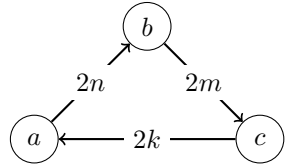
\begin{center}
\begin{tikzpicture}

\node[circle,draw, minimum width=0.25in] at (0,0) (a) {$a$}; 
\node[circle,draw,minimum width=0.25in] at (3,0) (c) {$c$}; 
\node[circle,draw,minimum width=0.25in] at (1.5,1.5) (b) {$b$}; 

\path[->,draw,thick] (b) to node[fill=white] {$2m$} (c);
\path[->,draw,thick] (c) to node[fill=white] {$2k$} (a);
\path[->,draw,thick] (a) to node[fill=white] {$2n$} (b);

\end{tikzpicture}
\end{center}
Let $\mathbf{P}'$ be obtained from $2\mathbf{P}$ by adding $2n+1$ voters with the ranking $cba$. Since $n<m$, $2m-2n-1>0$, so the margin graph of $\mathbf{P}'$ is
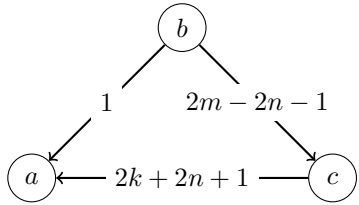
\begin{center}
\begin{tikzpicture}

\node[circle,draw, minimum width=0.25in] at (0,0) (a) {$a$}; 
\node[circle,draw,minimum width=0.25in] at (4,0) (c) {$c$}; 
\node[circle,draw,minimum width=0.25in] at (2,2) (b) {$b$}; 

\path[->,draw,thick] (b) to node[fill=white] {$2m-2n-1$} (c);
\path[->,draw,thick] (c) to node[fill=white] {$2k+2n+1$} (a);
\path[->,draw,thick] (b) to node[fill=white] {$1$} (a);

\end{tikzpicture}
\end{center}
By repeated application of positive involvement, $c\in F(\mathbf{P}')$. Yet $b$ is a Condorcet winner, so by the first paragraph of the proof, $F(\mathbf{P}')=\{b\}$, so we have a contradiction. Thus, $F(\mathbf{P})=\{b\}$, as desired.

Case 3: the margin graph of $\mathbf{P}$ is as in the \textit{Condorcet loser case}. In this case, $Minimax(\mathbf{P})=\{b,c\}$, and the argument that $F(\mathbf{P})\subseteq\{b,c\}$ is the same as in the first paragraph of the ascending case, only $\mathrm{Margin}_\mathbf{P}(a,b)$ is negative, which is irrelevant to the argument.

Case 4: the margin graph of $\mathbf{P}$ is as in the \textit{descending case}:
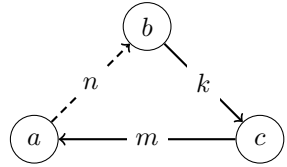
\begin{center}
\begin{tikzpicture}

\node[circle,draw, minimum width=0.25in] at (0,0) (a) {$a$}; 
\node[circle,draw,minimum width=0.25in] at (3,0) (c) {$c$}; 
\node[circle,draw,minimum width=0.25in] at (1.5,1.5) (b) {$b$}; 

\path[->,draw,thick] (b) to node[fill=white] {$k$} (c);
\path[->,draw,thick] (c) to node[fill=white] {$m$} (a);
\path[->,draw,thick, densely dashed] (a) to node[fill=white] {$n$} (b);

\end{tikzpicture}
\end{center}
where $n<m\leq k$. We claim that $F(\mathbf{P})=\{b\}$ in agreement with Minimax. Consider $8\mathbf{P}$:
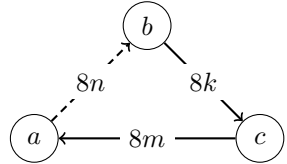
\begin{center}
\begin{tikzpicture}

\node[circle,draw, minimum width=0.25in] at (0,0) (a) {$a$}; 
\node[circle,draw,minimum width=0.25in] at (3,0) (c) {$c$}; 
\node[circle,draw,minimum width=0.25in] at (1.5,1.5) (b) {$b$}; 

\path[->,draw,thick] (b) to node[fill=white] {$8k$} (c);
\path[->,draw,thick] (c) to node[fill=white] {$8m$} (a);
\path[->,draw,thick,densely dashed] (a) to node[fill=white] {$8n$} (b);

\end{tikzpicture}
\end{center}

\noindent Let  $\ell = 4(k-m)$. Let $\mathbf{P}'$ be obtained from $8\mathbf{P}$ by adding $\ell+1$ blocks of all linear orders.  Let $\mathbf{P}''$ be obtained from $\mathbf{P}'$ by $1$ voter with the ranking $bac$ changing to $acb$, thereby increasing the margin of $a$ over $b$ by 2 and decreasing the margin of $b$ over $c$ by $2$. Let $\mathbf{P}'''$ be obtained from $\mathbf{P}''$ by removing $\ell$ voters with the ranking $bac$. Then the margin graph of $\mathbf{P}'''$ is:
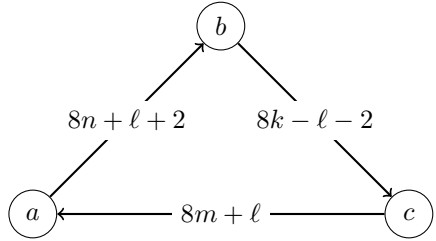
\begin{center}
\begin{tikzpicture}

\node[circle,draw, minimum width=0.25in] at (0,0) (a) {$a$}; 
\node[circle,draw,minimum width=0.25in] at (5,0) (c) {$c$}; 
\node[circle,draw,minimum width=0.25in] at (2.5,2.5) (b) {$b$}; 

\path[->,draw,thick] (b) to node[fill=white] {$8k-\ell-2$} (c);
\path[->,draw,thick] (c) to node[fill=white] {$8m+\ell$} (a);
\path[->,draw,thick] (a) to node[fill=white] {$8n+\ell+2$} (b);

\end{tikzpicture}
\end{center} 
We claim that $8n+\ell+2<8k-\ell-2<8m+\ell$, so the margin graph of $\mathbf{P}'''$ is a strictly ascending cycle.  Since $n<m$, we have  $n+1\leq m$, so $8(n+1)=8n+8\leq 8m$ and hence $8n+4 < 8m$, which implies $8n+4k - 4m +2<8k-4k + 4m-2$, which is the first inequality in our claim. The second inequality is simply $8k-4k + 4m - 2 < 8m + 4k - 4m$, which is immediate.

Thus, the margin graph of $\mathbf{P}'''$ is a strictly ascending cycle, so by Case~2, $F(\mathbf{P}''')=\{b\}$, so by $\ell$ applications of positive involvement, $b\in F(\mathbf{P}'')$, so by weak positive responsiveness, $F(\mathbf{P}')=\{b\}$, so by $\ell+1$ applications of block preservation, $F(8\mathbf{P})=\{b\}$, so by homogeneity, $F(\mathbf{P})=\{b\}$, in line with Minimax.\end{proof}

\begin{remark}\label{Domains} In Section \ref{Prelim}, we assumed in the tradition of Arrow \citeyearpar{Arrow1963} that the domain of a voting method is the set of profiles of strict weak orders. However, inspection of the proofs of Facts \ref{MinimaxMay}, \ref{MinimaxPI}, \ref{MinimaxIS}, and \ref{MinimaxHBP} and of Theorem \ref{MainThm} shows that for these results, we do not actually need to assume any properties of a voter's binary relation $\mathbf{P}(i)$ except that it is a strict preference relation, i.e., it is asymmetric.\footnote{When we say that $i$ ``ranks $a$ uniquely first'', we simply mean that $(a,b)\in\mathbf{P}(i)$ for all $b\in X(\mathbf{P})\setminus\{a\}$ (and similarly for ``ranks $a$ uniquely last'' but with $(b,a)\in\mathbf{P}(i)$), not that $\mathbf{P}(i)$ forms a transitive ranking.} In particular, all of these results also apply to the setting of linear ballots, where ties are not allowed, to the setting of truncated ballots, where some alternatives may be unranked (regarded as tied for last), and to the setting where instead of asking voters for a ranking, we simply ask for their pairwise choices (Do you prefer $a$ or $b$? Do you prefer $a$ or $c$? Do you prefer $b$ or $c$?), which allows for intransitivities in $\mathbf{P}(i)$. This latter point is especially important given the recent proposal of Foley and Maskin \citeyearpar{Foley2024} that in top-three general elections, we should ask voters to answer three pairwise comparison questions rather than to provide a ranking of the three alternatives. Theorem \ref{MainThm} continues to apply in this setting.\end{remark}

Thus, for three-alternative profiles, if in addition to May's axioms we want the axioms of positive involvement and immunity to spoilers (or just near immunity), plus the relatively uncontroversial homogeneity and block preservation axioms, then we must refine Minimax. Note that in the proof of Theorem \ref{MainThm}, the only use of near immunity to spoilers is to help derive (together with other axioms) that $F$ satisfies Condorcet consistency. Hence we have the following reformulation of Theorem \ref{MainThm}.

\begin{proposition}\label{CondorcetVersion} Theorem \ref{MainThm} holds with near immunity to spoilers replaced by Condorcet consistency.
\end{proposition}

\begin{remark}\label{SupportVMargin} Theorem \ref{MainThm} bears on a debate about how Minimax should be formulated. When an alternative $b$ has a positive margin over $a$, should we measure the size of $a$'s loss by the \textit{margin} of $b$ over $a$, as we have, or rather by the \textit{support} for $b$ over $a$, defined as the absolute number of voters who rank $b$ above $a$ (cf.~\citealt[p.~269]{Schulze2011})? If voters' rankings are all strict linear orders, the choice does not matter, but in profiles allowing ties (even if only at the bottom of voters' rankings), the two approaches are not equivalent. In fact, Minimax formulated in terms of support rather than margin violates positive involvement. Consider the following profile (where the number above each ranking indicates the number of voters with that ranking, and `$b,c$' indicates a tie) with its margin graph (center) and support graph (right):

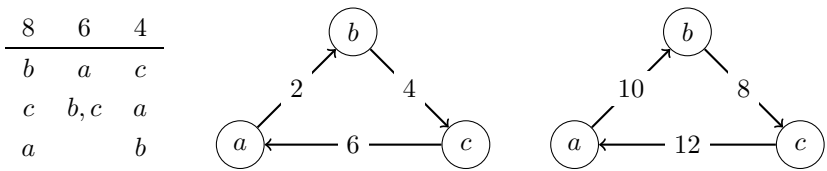
\begin{center}
\begin{minipage}{1.1in}
\begin{tabular}{ccc}
 8 & 6  & 4   \\ \hline 
 $b$ & $a$ & $c$   \\
  $c$ & $b,c$  & $a$   \\
   $a$ &  & $b$       
\end{tabular}\end{minipage}\begin{minipage}{1.75in}\begin{tikzpicture}

\node[circle,draw, minimum width=0.25in] at (0,0) (a) {$a$}; 
\node[circle,draw,minimum width=0.25in] at (3,0) (c) {$c$}; 
\node[circle,draw,minimum width=0.25in] at (1.5,1.5) (b) {$b$}; 

\path[->,draw,thick] (b) to node[fill=white] {$4$} (c);
\path[->,draw,thick] (c) to node[fill=white] {$6$} (a);
\path[->,draw,thick] (a) to node[fill=white] {$2$} (b);

\end{tikzpicture}\end{minipage}\begin{minipage}{1.75in}\begin{tikzpicture}

\node[circle,draw, minimum width=0.25in] at (0,0) (a) {$a$}; 
\node[circle,draw,minimum width=0.25in] at (3,0) (c) {$c$}; 
\node[circle,draw,minimum width=0.25in] at (1.5,1.5) (b) {$b$}; 

\path[->,draw,thick] (b) to node[fill=white] {$8$} (c);
\path[->,draw,thick] (c) to node[fill=white] {$12$} (a);
\path[->,draw,thick] (a) to node[fill=white] {$10$} (b);

\end{tikzpicture}\end{minipage}
\end{center}
Support-based Minimax says that $c$ has the ``smallest loss'', so $c$ wins. But then as in the strictly ascending case of the proof of Theorem \ref{MainThm}, adding $3$ voters with the ranking $cba$ makes $b$ the Condorcet winner and hence the support-based Minimax winner, violating positive involvement. We take this to be an argument in favor of the margin-based formulation of Minimax instead of the support-based formulation.\end{remark}

Finally, we note that there are indeed \textit{proper} refinements of Minimax satisfying the axioms in Theorem~\ref{MainThm}. Since the issue of how to break Minimax ties is much less important than constraining a voting method to choose from among the Minimax winners, we defer discussion of such refinements to the Appendix. However, it is noteworthy that with the refinement defined in the Appendix, we satisfy not only weak positive responsiveness but even full positive responsiveness as in Remark \ref{PRremark}. Yet Proposition \ref{AppendixProp} in the Appendix shows that there is a tradeoff between satisfying positive responsiveness (rather than just weak positive responsiveness) and immunity to spoilers (rather than just near immunity to spoilers).

\section{Independence of axioms}\label{IndependenceSection}

In this section, we turn to the question of the independence of the axioms in Theorem \ref{MainThm}. First, however, we claim that we can drop May's axiom of anonymity from our list of axioms. Anonymity is only used in the proof of Theorem~\ref{MainThm} to help derive that the voting method refines majority voting in two-alternative profiles. But we can derive this result without anonymity as follows.

\begin{proposition}\label{DropAnon} If $F$ satisfies neutrality, weak positive responsiveness, homogeneity, and block preservation, then $F$ agrees with majority voting in any two-alternative profile with a unique majority winner.
\end{proposition}
\begin{proof} Suppose $\mathbf{P}$ is a two-alternative profile with a unique majority winner. Hence there are $a,b\in X(\mathbf{P})$ such that $n$ voters rank $b$ above $a$, $n+k$ (for $k>0$) rank $a$ above $b$, and the remaining $\ell$ voters rank them as tied. Let $\mathbf{Q}=4\mathbf{P}$. If we can show that $F(\mathbf{Q})=\{a\}$, then by homogeneity, $F(\mathbf{P})=\{a\}$, in agreement with majority voting, as desired.

Let $\alpha=4n$, $\beta=4k$, and $\gamma=4\ell$. Since $k>0$, we have $\beta \geq 4$. In $\mathbf{Q}$, $\alpha$ voters rank $b$ above $a$, $\alpha+\beta$ rank $a$ above $b$, and the remaining $\gamma$ voters rank them as tied. Consider the profile $\mathbf{Q}_0$ consisting of just $\beta-2$ of the voters from $\mathbf{Q}$ who rank $a$ above $b$, still ranking $a$ above $b$, and all of the $\gamma$ voters from $\mathbf{Q}$ who ranked them as tied, still ranking them as tied. We claim that $F(\mathbf{Q}_0)=\{a\}$. For if $b\in F(\mathbf{Q}_0)$, then where $\mathbf{R}_0$ is obtained from $\mathbf{Q}_0$ by the $\beta-2$ voters who ranked $a$ above $b$ switching to rank $b$ above $a$, weak positive responsiveness applied $\beta-2$ times implies $F(\mathbf{R}_0)=\{b\}$, while neutrality together with $b\in F(\mathbf{Q}_0)$ implies $a\in F(\mathbf{R}_0)$, a contradiction. Hence $F(\mathbf{Q}_0)=\{a\}$. Now observe that there are exactly $\alpha+2$ voters from $\mathbf{Q}$ who ranked $a$ above $b$ who are not in $\mathbf{Q}_0$. Let $\mathbf{Q}_1$ be obtained from $\mathbf{Q}_0$ by adding all of those $\alpha+2$ voters, where $\alpha+1$ of them still rank $a$ above $b$ but one of them---say, voter~$i_1$---now ranks $b$ above $a$, and adding all $\alpha$ of the original voters who ranked $b$ above $a$, still ranking $b$ above $a$. Then $\mathbf{Q}_1$ is obtained from $\mathbf{Q}_0$ by adding $\alpha+1$ pairs of the rankings $ab$ and $ba$, so by block preservation, $F(\mathbf{Q}_0)=\{a\}$ implies $a\in F(\mathbf{Q}_1)$. Finally, observe that $\mathbf{Q}$ can be obtained from $\mathbf{Q}_1$ by voter $i_1$ moving $a$ from last place to first place, so $a\in F(\mathbf{Q}_1)$ implies $F(\mathbf{Q})=\{a\}$ by weak positive responsiveness, in agreement with majority voting. The relations between $\mathbf{Q}$, $\mathbf{Q}_0$, and $\mathbf{Q}_1$ are shown in Figure \ref{DropAnonFig}.
\end{proof}

As an immediate corollary of Proposition \ref{DropAnon}, we obtain the following.

\begin{proposition} Theorem \ref{MainThm} holds without the assumption of anonymity.
\end{proposition}

\begin{figure}
\begin{center}

$\mathbf{Q}$\quad \begin{tabular}{c|c|c|c|c|c|c|c|c}
$h_1$ & $\dots$ &  $ h_{\alpha}$ & $i_1$  & $\dots$ & $i_{\alpha+\beta}$ & $j_1$ & $\dots $ & $j_\gamma$  \\
\hline 
$b$ & $\dots$ & $b$ & $a$ & $\dots$   & $a$ & $a,b$ & $\dots$ & $a,b$  \\
$a$ & $\dots$ & $a$ & $b$ & $\dots$ & $b$ &  & & 
\end{tabular}\vspace{.2in}

$\mathbf{Q}_0$\quad \begin{tabular}{c|c|c|c|c|c}
 $i_{\alpha+3}$ & $\dots$ & $i_{\alpha+\beta}$ & $j_1$ & $\dots $ & $j_\gamma$  \\
\hline 
 $a$ & $\dots$ & $a$ & $a,b$ & $\dots$ & $a,b$  \\
 $b$ & $\dots$ & $b$ & & &
\end{tabular}\vspace{.2in}

$\mathbf{Q}_1$\quad \begin{tabular}{c|c|c|c|c|c|c|c|c|c|c|c|c}
$h_1$ & $\dots$ &  $ h_{\alpha}$ & $\boldsymbol{i_1}$ & $i_2$ & $\dots$ & $i_{\alpha+2}$ & $i_{\alpha+3}$ & $\dots$ & $i_{\alpha+\beta}$ & $j_1$ & $\dots $ & $j_\gamma$  \\
\hline 
$b$ & $\dots$ & $b$ & $\boldsymbol{b}$ & $a$ & $\dots$ & $a$ & $a$ & $\dots$ & $a$ & $a,b$ & $\dots$ & $a,b$  \\
$a$ & $\dots$ & $a$ & $\boldsymbol{a}$ & $b$ & $\dots$ & $b$ & $b$ & $\dots$ & $b$ & & & \\
 \multicolumn{4}{c}{\upbracefill}&  \multicolumn{3}{c}{\upbracefill} &  \multicolumn{6}{c}{\upbracefill}\\
 \multicolumn{4}{c}{$\scriptstyle \alpha+1$ {\footnotesize voters}}&  \multicolumn{3}{c}{$\scriptstyle \alpha+1$ {\footnotesize voters}} & \multicolumn{6}{c}{{\footnotesize voters from }$\scriptstyle \mathbf{Q}_0$ }
\end{tabular}

\end{center}
\caption{profiles for the proof of Proposition \ref{DropAnon}. Rankings are labeled by voter names above the rankings. Higher alternatives are ranked above lower alternatives, while $a,b$ on the same line represents a tie.}\label{DropAnonFig}
\end{figure}

We now show that each axiom in Theorem~\ref{MainThm} is independent of the other axioms in the following sense.

\begin{proposition}\label{Independence} For each of the axioms $A$ in Theorem \ref{MainThm}, there is a voting method $F_A$ on the domain of profiles with two or three alternatives such that $F_A$ satisfies all of the axioms in Theorem~\ref{MainThm} except for $A$, which it violates; and if $A$ is any of the axioms other than anonymity, then $F_A$ does not refine Minimax.\end{proposition}

We prove Proposition \ref{Independence} in pieces to highlight which cases require more work.

\begin{lemma}\label{ThreeAxLem} Proposition \ref{Independence} holds for \[A\in \{\mbox{anonymity, neutrality, weak positive responsiveness, near  immunity to spoilers}\}.\]
\end{lemma}

\begin{proof} Where $A$ is anonymity, fix designated voters $i,j\in\mathcal{V}$ and define a voting method $F$ as follows: if $V(\mathbf{P})=\{i,j\}$ and $i$ and $j$ have reversed linear orders of the alternatives in $\mathbf{P}$, then $F(\mathbf{P})$ contains only $i$'s favorite alternative; otherwise $F(\mathbf{P})=Minimax(\mathbf{P})$. Clearly $F$ violates anonymity, but it is easy to check that it satisfies all the other axioms.

Where $A$ is neutrality, fixing a well-order $\leqslant$ on $\mathcal{X}$, let $F(\mathbf{P})$ contain only the alternative from $X(\mathbf{P})$ that comes first according to $\leqslant$. Clearly $F$ violates neutrality, but it is easy to check that $F$ satisfies all the other axioms (for immunity to spoilers, if $a$ comes first in $\{a,c\}$ and $\{a,b\}$, then  $a$ comes first in $\{a,b,c\}$).

Where $A$ is weak positive responsiveness, let $F$ be the voting method such that $F(\mathbf{P})=X(\mathbf{P})$. Clearly $F$ violates weak positive responsiveness, but it is easy to check that $F$ satisfies all the other axioms.

Where $A$ is near immunity to spoilers, the Borda voting method clearly satisfies all the axioms except for near immunity to spoilers, which it violates in three-alternative profiles (see \citealt[Example 4.15]{HP2021}); it is also easy to construct three-alternative profiles with $Borda(\mathbf{P})\not\subseteq Minimax(\mathbf{P})$.
\end{proof}

\begin{lemma}\label{PIind} Proposition \ref{Independence} holds for $A=\mbox{positive involvement}$.
\end{lemma}

\begin{proof} Given a profile $\mathbf{P}$, an alternative $a$ is a \textit{weak Condorcet winner} if $a$ has a non-negative margin over each of the other alternatives in $\mathbf{P}$. The \textit{plurality score} of an alternative $a$ in $\mathbf{P}$ is the number of voters who rank $a$ uniquely first. Now consider the (weak) Condorcet-Plurality voting method $\mbox{\textit{CP}}$: if there is a weak Condorcet winner in $\mathbf{P}$, then $\mbox{\textit{CP}}(\mathbf{P})$ is the set of weak Condorcet winners; otherwise $\mbox{\textit{CP}}(\mathbf{P})$ is the set of alternatives with maximal plurality score.\footnote{Support-based Minimax as in Remark \ref{SupportVMargin} can also be used to prove the independence of positive involvement, but we like the example of Condorcet-Plurality as a method further away from margin-based Minimax.}  To see that $CP$ violates positive involvement, consider the following profiles $\mathbf{P}$ and $\mathbf{P}'$ with their respective margin graphs:
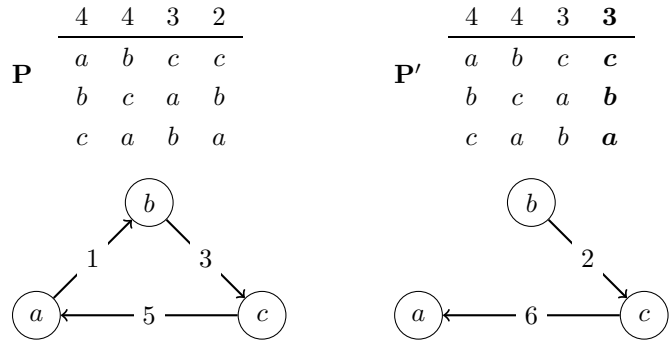
\begin{center}
\begin{minipage}{2in}
$\mathbf{P}$\quad\begin{tabular}{cccc}
$4$ & $4$ & $3$ & $2$ \\
\hline
$a$ & $b$ & $c$ & $c$  \\
$b$ & $c$ & $a$ & $b$ \\
$c$ & $a$ & $b$ & $a$ 
\end{tabular}\vspace{.1in}
\begin{tikzpicture}

\node[circle,draw, minimum width=0.25in] at (0,0) (a) {$a$}; 
\node[circle,draw,minimum width=0.25in] at (3,0) (c) {$c$}; 
\node[circle,draw,minimum width=0.25in] at (1.5,1.5) (b) {$b$}; 

\path[->,draw,thick] (b) to node[fill=white] {$3$} (c);
\path[->,draw,thick] (c) to node[fill=white] {$5$} (a);
\path[->,draw,thick] (a) to node[fill=white] {$1$} (b);

\end{tikzpicture} 
\end{minipage}\begin{minipage}{2in}
$\mathbf{P}'$\quad\begin{tabular}{cccc}
$4$ & $4$ & $3$ & $\boldsymbol{3}$ \\
\hline
$a$ & $b$ & $c$ & $\boldsymbol{c}$  \\
$b$ & $c$ & $a$ & $\boldsymbol{b}$ \\
$c$ & $a$ & $b$ & $\boldsymbol{a}$ 
\end{tabular}\vspace{.1in}
\begin{tikzpicture}

\node[circle,draw, minimum width=0.25in] at (0,0) (a) {$a$}; 
\node[circle,draw,minimum width=0.25in] at (3,0) (c) {$c$}; 
\node[circle,draw,minimum width=0.25in] at (1.5,1.5) (b) {$b$}; 

\path[->,draw,thick] (b) to node[fill=white] {$2$} (c);
\path[->,draw,thick] (c) to node[fill=white] {$6$} (a);

\end{tikzpicture}
\end{minipage}
\end{center}
In $\mathbf{P}$, there is no Condorcet winner, so $\mbox{\textit{CP}}(\mathbf{P})=\{c\}$ since $c$ has the maximum plurality score of $5$. Then $\mathbf{P}'$ is obtained from $\mathbf{P}$ by adding one voter with the ranking $cba$, but in $\mathbf{P'}$, $b$ is the unique weak Condorcet winner, so $\mbox{\textit{CP}}(\mathbf{P}')=\{b\}$, violating positive involvement.

Clearly $CP$ satisfies anonymity, neutrality, homogeneity, and block preservation.  For weak positive responsiveness, suppose $a\in \mbox{\textit{CP}}(\mathbf{P})$ and $\mathbf{P}'$ is obtained from $\mathbf{P}$ by one voter in $V(\mathbf{P})$ who ranked $a$ uniquely last in $\mathbf{P}$ switching to ranking $a$ uniquely first in $\mathbf{P}'$, otherwise keeping their ranking the same. 

Case 1: there is a weak Condorcet winner in $\mathbf{P}$. Then $a\in \mbox{\textit{CP}}(\mathbf{P})$ implies that $a$ is a weak Condorcet winner in $\mathbf{P}$. Then since $a$'s margins against all other alternatives increase from $\mathbf{P}$ to $\mathbf{P}'$, $a$ is the Condorcet winner and hence the only weak Condorcet winner in $\mathbf{P}'$, so $\mbox{\textit{CP}}(\mathbf{P}')=\{a\}$. 

Case 2: there is no weak Condorcet winner in $\mathbf{P}$. Then $a\in \mbox{\textit{CP}}(\mathbf{P})$ implies that $a$ is among the alternatives with maximal plurality score in $\mathbf{P}$. Then clearly $a$ has the unique maximum plurality score in $\mathbf{P}'$, and if moving from $\mathbf{P}$ to $\mathbf{P}'$ creates a weak Condorcet winner, it must be $a$ alone, so in either case we have $\mbox{\textit{CP}}(\mathbf{P}')=\{a\}$. 

To see that $CP$ satisfies immunity to spoilers in profiles with two or three alternatives, suppose $a\in \mbox{\textit{CP}}(\mathbf{P}_{-b})$ and $\mbox{\textit{CP}}(\mathbf{P}_{\mid \{a,b\}})=\{a\}$. Since $\mathbf{P}$ has at most three alternatives, $a\in \mbox{\textit{CP}}(\mathbf{P}_{-b})$ implies that $a$ is a weak Condorcet winner in $\mathbf{P}_{-b}$, and $\mbox{\textit{CP}}(\mathbf{P}_{\mid \{a,b\}})=\{a\}$ implies $a$ is the Condorcet winner in $\mathbf{P}_{\mid \{a,b\}}$. It follows that $a$ is a weak Condorcet winner in $\mathbf{P}$, so $a\in \mbox{\textit{CP}}(\mathbf{P})$. \end{proof}

In the following, the voting methods showing the independence of homogeneity and block preservation are highly contrived. We suspect there are no natural voting methods besides refinements of Minimax that satisfy May's axioms, positive involvement, and near immunity to spoilers.

\begin{lemma}  Proposition \ref{Independence} holds for $A=\mbox{homogeneity}$.
\end{lemma}

\begin{proof} Consider the voting method $F$ defined on profiles $\mathbf{P}$ with two or three alternatives as follows: $F(\mathbf{P})=Minimax(\mathbf{P})$ except if (I) each alternative has a unique Minimax score, (II) the second best Minimax score differs from the best Minimax score by at most 1, and (III) the alternative with the second best Minimax score beats the Minimax winner head-to-head, in which case $F(\mathbf{P})$ contains both the Minimax winner and the alternative that beats the Minimax winner head-to-head. For example, in the following profile $\mathbf{P}$, we have $F(\mathbf{P})=\{a,c\}$:
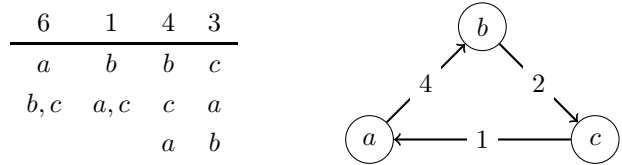
\begin{center}
\begin{minipage}{1.75in}
\begin{tabular}{cccc}
 6 & 1  & 4 & 3   \\ \hline 
 $a$ & $b$ & $b$ & $c$   \\
  $b,c$ & $a,c$  & $c$ & $a$  \\
    &  & $a$   & $b$
\end{tabular}\end{minipage}\begin{minipage}{1.75in}\begin{tikzpicture}

\node[circle,draw, minimum width=0.25in] at (0,0) (a) {$a$}; 
\node[circle,draw,minimum width=0.25in] at (3,0) (c) {$c$}; 
\node[circle,draw,minimum width=0.25in] at (1.5,1.5) (b) {$b$}; 

\path[->,draw,thick] (b) to node[fill=white] {$2$} (c);
\path[->,draw,thick] (c) to node[fill=white] {$1$} (a);
\path[->,draw,thick] (a) to node[fill=white] {$4$} (b);

\end{tikzpicture}\end{minipage}
\end{center}
Observe that in $2\mathbf{P}$, the difference between the best and second best Minimax scores is 2, so $F(2\mathbf{P})=\{a\}$, violating homogeneity. Yet clearly $F$ satisfies anonymity, neutrality, block preservation (since adding a block of all linear orders does not change margins), and immunity to spoilers (since $F$ agrees with Minimax in two-alternative profiles and selects all of the Minimax winners in three-alternative profiles). It remains to show that $F$ satisfies weak positive responsiveness and positive involvement.

For positive involvement, for an arbitrary profile $\mathbf{P}$, suppose  $x\in F(\mathbf{P})$ and that $\mathbf{P}'$ is obtained from $\mathbf{P}$ by adding one voter who ranks $x$ first. If  $x\in Minimax(\mathbf{P})$, then $x \in Minimax (\mathbf{P}')$ by positive involvement for Minimax, so $x\in F(\mathbf{P}')$ by the definition of $F$. On the other hand, if conditions (I)-(III) hold and $x$ is the alternative with the second best Minimax score in $\mathbf{P}$, then since $x$ beats the Minimax winner $m$ head-to-head, $m$'s Minimax score is given by $x$'s margin over $m$ (since otherwise $m$ would have two losses, which is impossible for a Minimax winner in a three-alternative profile), so  in $\mathbf{P}'$ the Minimax score of $m$ increases by 1, while the Minimax score of $x$ decreases by $1$, so $x\in Minimax(\mathbf{P}')\subseteq F(\mathbf{P}')$.

For weak positive responsiveness, suppose $x\in F(\mathbf{P})$, and $\mathbf{P}'$ is obtained from $\mathbf{P}$ by  one voter moving $x$ from last place to first place. Since $x\in F(\mathbf{P})$, one of the following cases obtains.

Case 1: $x$'s Minimax score is minimal in $\mathbf{P}$. Then since $x$'s Minimax score decreases by 2 from $\mathbf{P}$ to $\mathbf{P}'$, while the scores of the other alternatives do not decrease, condition (II) does not obtain in $\mathbf{P}'$, so $F(\mathbf{P}')=Minimax(\mathbf{P}')=\{x\}$.

Case 2: $x$'s Minimax score is 1 greater than the minimal score in $\mathbf{P}$, and $x$ beats the Minimax winner head-to-head.  Let $y$ be the Minimax winner and $z$ the alternative that beats $x$ head-to-head in $\mathbf{P}$. By condition (I), the difference between the Minimax scores of $y$ and $z$ in $\mathbf{P}$ is at least 2. Then since $y$'s Minimax score increases by only 2 from $\mathbf{P}$ to $\mathbf{P}'$, while $z$'s does not decrease, $y$'s Minimax score is no greater than $z$'s Minimax score in $\mathbf{P}'$. Hence either $y$ and $z$ have the same Minimax score in $\mathbf{P}'$, in which case condition (I) does not hold for $\mathbf{P}'$, or $y$'s Minimax score is smaller than $z$'s, in which case condition (III) does not obtain in $\mathbf{P}'$, since $y$ does not beat $x$ head-to-head in $\mathbf{P}'$.  In either case, $F(\mathbf{P}')=Minimax(\mathbf{P}')=\{x\}$.\end{proof}

\begin{lemma}\label{BPLem} Proposition \ref{Independence} holds for $A=\mbox{block preservation}$.
\end{lemma}
\begin{proof} Consider the voting method $F$ defined on profiles $\mathbf{P}$ with two or three alternatives as follows: $F(\mathbf{P})=Minimax(\mathbf{P})$ except if (i) the margin graph of $\mathbf{P}$ is a cycle of the form
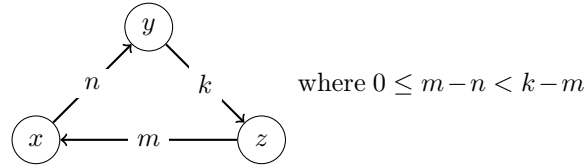
\begin{center}
\begin{minipage}{1.5in}\begin{tikzpicture}

\node[circle,draw, minimum width=0.25in] at (0,0) (a) {$x$}; 
\node[circle,draw,minimum width=0.25in] at (3,0) (c) {$z$}; 
\node[circle,draw,minimum width=0.25in] at (1.5,1.5) (b) {$y$}; 

\path[->,draw,thick] (b) to node[fill=white] {$k$} (c);
\path[->,draw,thick] (c) to node[fill=white] {$m$} (a);
\path[->,draw,thick] (a) to node[fill=white] {$n$} (b);

\end{tikzpicture}\end{minipage}\begin{minipage}{1.5in}where $0\leq m-n < k-m$\end{minipage}
\end{center}
and (ii) no voter submits the ranking $yxz$, in which case $F(\mathbf{P})=\{x\}$. Note that we are using `$x$', `$y$', and `$z$' here as variables, so if there is any assignment of the alternatives in $\mathbf{P}$ to the roles $x,y,z$ such that (i) and (ii) hold, then $F$ selects the alternative playing the $x$ role, rather than the Minimax winner. That there are profiles satisfying (i) and (ii), so $F$ differs from Minimax, is shown by the following, where $F(\mathbf{P})=\{x\}$:
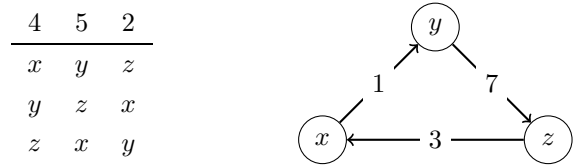
\begin{center}
\begin{minipage}{1.5in}
\begin{tabular}{ccc}
4 & 5  & 2   \\ \hline 
 $x$ & $y$ & $z$   \\
  $y$ & $z$  & $x$   \\
   $z$ & $x$ & $y$       
\end{tabular}\end{minipage}\begin{minipage}{1.5in}\begin{tikzpicture}

\node[circle,draw, minimum width=0.25in] at (0,0) (a) {$x$}; 
\node[circle,draw,minimum width=0.25in] at (3,0) (c) {$z$}; 
\node[circle,draw,minimum width=0.25in] at (1.5,1.5) (b) {$y$}; 

\path[->,draw,thick] (b) to node[fill=white] {$7$} (c);
\path[->,draw,thick] (c) to node[fill=white] {$3$} (a);
\path[->,draw,thick] (a) to node[fill=white] {$1$} (b);

\end{tikzpicture}\end{minipage}
\end{center}
\noindent  Adding a block of all linear orders to the profile above---and hence adding the ranking $yxz$---changes the $F$ winner to be the Minimax winner, $y$, so $F$ violates block preservation. Yet clearly $F$ satisfies anonymity, neutrality, homogeneity,  and immunity to spoilers (since $F$ agrees with Minimax in all of the profiles satisfying the assumptions in the definition of immunity to spoilers). It remains to prove positive involvement and weak positive responsiveness. Consider an arbitrary profile~$\mathbf{P}$.

Case 1: there are weak Condorcet winners in $\mathbf{P}$. Then $F(\mathbf{P})=Minimax(\mathbf{P})$, and for $a\in F(\mathbf{P})$, adding a voter who ranks $a$  first  (resp.~having some voter move $a$ from last to first) makes $a$ a Condorcet winner, so $F(\mathbf{P}')=Minimax(\mathbf{P}')=\{a\}$, in line with positive involvement (resp.~weak positive responsiveness).

Case 2: there are no weak Condorcet winners in $\mathbf{P}$, so there is a cycle.

Case 2.1: condition (i) in the definition of $F$ holds. Suppose alternatives $a$, $b$, and $c$ play the roles of $x$, $y$, and $z$ in condition (i):
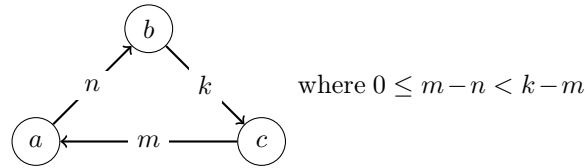
\begin{center}
\begin{minipage}{1.5in}\begin{tikzpicture}

\node[circle,draw, minimum width=0.25in] at (0,0) (a) {$a$}; 
\node[circle,draw,minimum width=0.25in] at (3,0) (c) {$c$}; 
\node[circle,draw,minimum width=0.25in] at (1.5,1.5) (b) {$b$}; 

\path[->,draw,thick] (b) to node[fill=white] {$k$} (c);
\path[->,draw,thick] (c) to node[fill=white] {$m$} (a);
\path[->,draw,thick] (a) to node[fill=white] {$n$} (b);

\end{tikzpicture}\end{minipage}\begin{minipage}{1.5in}where $0\leq m-n < k-m$\end{minipage}
\end{center}
In this case, we need a lemma: if $a\in F(\mathbf{P})$, $\mathbf{P}'$ is obtained from $\mathbf{P}$ by adding one voter who ranks $a$ in first place or having some voter move $a$ from last to first place, and as a result condition (i) no longer holds in $\mathbf{P}'$ with $a\mapsto x$, $b\mapsto y$, and $c\mapsto z$, then $F(\mathbf{P}')=\{a\}$. To see this, suppose $\mathbf{P}$ and $\mathbf{P}'$ are as described. Since condition (i) does not hold in $\mathbf{P}'$ with $a\mapsto x$, $b\mapsto y$, and $c\mapsto z$,  either $0\not\leq Margin_{\mathbf{P}'}(c,a) - Margin_{\mathbf{P}'}(a,b)$ or $ Margin_{\mathbf{P}'}(c,a) - Margin_{\mathbf{P}'}(a,b) \not<Margin_{\mathbf{P}'}(b,c) - Margin_{\mathbf{P}'}(c,a) $. But we can rule out the latter possibility, since by the assumption of the case, we have $ 0\leq Margin_{\mathbf{P}}(c,a) - Margin_{\mathbf{P}}(a,b) <Margin_{\mathbf{P}}(b,c) - Margin_{\mathbf{P}}(c,a) $, and the margin of $a$ over $b$ increases from $\mathbf{P}$ to $\mathbf{P}'$ by at least as much as the margin of $b$ over $c$ decreases from $\mathbf{P}$ to $\mathbf{P}'$, while the margin of $c$ over $a$ decreases. Thus, we conclude that $0\not\leq Margin_{\mathbf{P}'}(c,a) - Margin_{\mathbf{P}'}(a,b)$, i.e., $ Margin_{\mathbf{P}'}(c,a) < Margin_{\mathbf{P}'}(a,b) $. Note that since $ Margin_{\mathbf{P}}(c,a) < Margin_{\mathbf{P}}(b,c)$ and the margin of $c$ over $a$ decreases by at least as much as the margin of $b$ over $c$ decreases from $\mathbf{P}$ to $\mathbf{P}'$, we have $ Margin_{\mathbf{P}'}(c,a) < Margin_{\mathbf{P}'}(b,c)$. Now there are two cases:

First, suppose $ Margin_{\mathbf{P}'}(c,a) < Margin_{\mathbf{P}'}(a,b)\leq Margin_{\mathbf{P}'}(b,c)$. Thus, following the cycle from its weakest edge, we increase and then weakly increase in margin. It follows that the margin graph of $\mathbf{P}'$ does not satisfy condition (i) under any bijection from $\{a,b,c\}$ to $\{x,y,z\}$. Hence~$F(\mathbf{P}')=Minimax(\mathbf{P}')=\{a\}$.

Second, suppose $ Margin_{\mathbf{P}'}(c,a) < Margin_{\mathbf{P}'}(b,c) < Margin_{\mathbf{P}'}(a,b)$. Thus, the only way $\mathbf{P}'$ can satisfy (i) is with $a\mapsto y$, $b\mapsto z$, $c\mapsto x$. But we claim that the required inequality in (i) is not satisfied under this bijection. Since $Margin_{\mathbf{P}}(b,c) - Margin_{\mathbf{P}}(a,b)\geq 1$, it follows by how $\mathbf{P}'$ was obtained from $\mathbf{P}$ that $Margin_{\mathbf{P}'}(a,b)-Margin_{\mathbf{P}'}(b,c) \leq 1$. But then $Margin_{\mathbf{P}'}(b,c)  - Margin_{\mathbf{P}'}(c,a)  \not < Margin_{\mathbf{P}'}(a,b)-Margin_{\mathbf{P}'}(b,c)$, since the left side is at least $1$ and the right side at most $1$. It follows that the margin graph of $\mathbf{P}'$ does not satisfy condition (i) under the stated bijection.  Hence $F(\mathbf{P}')=Minimax(\mathbf{P}')=\{a\}$.

With the lemma established, we now proceed with our main case analysis.

Case 2.1.1: condition (ii)  holds, so there is no voter with the ranking $bac$. Then $F(\mathbf{P})=\{a\}$ (contrary to Minimax). Now if we add one voter who ranks $a$ in first place or have some voter move $a$ from last to first place to obtain $\mathbf{P}'$, then we have two cases.

Case 2.1.1.1: the margin graph of $\mathbf{P}'$ is still as in (i) with $a\mapsto x$, $b\mapsto y$, and $c\mapsto z$. Then since (ii) still holds in $\mathbf{P}'$, we still have $F(\mathbf{P}')=\{a\}$.

Case 2.1.1.2: the margin graph of $\mathbf{P}'$ is no longer as in (i) with $a\mapsto x$, $b\mapsto y$, and $c\mapsto z$.  Then by the lemma established above, we again have $F(\mathbf{P}')=\{a\}$.

Case 2.1.2: condition (ii) does not hold, i.e., there is some voter with the ranking $bac$.

Case 2.1.2.1: $n<m$. Then $F(\mathbf{P})=Minimax(\mathbf{P})=\{b\}$. If we add a voter who ranks $b$ in first place or have some voter move $b$ from last  to first place, then condition (i) still holds for $\mathbf{P}'$ with $a\mapsto x$, $b\mapsto y$, and $c\mapsto z$, and since there is still a voter with $bac$, condition (ii) still does not hold for $\mathbf{P}'$, so $F(\mathbf{P}')=Minimax(\mathbf{P}')=\{b\}$.

Case 2.1.2.2: $n=m$. Then $F(\mathbf{P})=Minimax(\mathbf{P})=\{a,b\}$. If we add a voter who ranks $b$ in first place or have some voter move $b$ from last to first, then condition (i) still holds for $\mathbf{P}'$ with $a\mapsto x$, $b\mapsto y$, and $c\mapsto z$, and  condition (ii) still does not hold for $\mathbf{P}'$, so $F(\mathbf{P}')=Minimax(\mathbf{P}')=\{b\}$.  Now suppose we add a voter who ranks $a$ in first place or have some voter move $a$ from last to first. Then since $n=m$, in $\mathbf{P}'$ the margin of $c$ over $a$ is smaller than that of $a$ over $b$ and that of $b$ over $c$, so condition (i) no longer holds in $\mathbf{P}'$ with $a\mapsto x$, $b\mapsto y$, and $c\mapsto z$. Therefore, by the lemma established above, $F(\mathbf{P}')=\{a\}$.

Case 2.2: condition (i) in the definition of $F$ does not hold, so we have a cycle of the form 
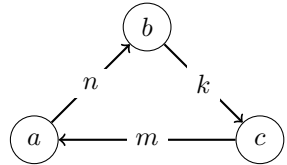
\begin{center}
\begin{tikzpicture}

\node[circle,draw, minimum width=0.25in] at (0,0) (a) {$a$}; 
\node[circle,draw,minimum width=0.25in] at (3,0) (c) {$c$}; 
\node[circle,draw,minimum width=0.25in] at (1.5,1.5) (b) {$b$}; 

\path[->,draw,thick] (b) to node[fill=white] {$k$} (c);
\path[->,draw,thick] (c) to node[fill=white] {$m$} (a);
\path[->,draw,thick] (a) to node[fill=white] {$n$} (b);

\end{tikzpicture}
\end{center}
where either $n=m=k$ or $n<m$, $n<k$, and $k-m\leq m-n$.

Case 2.2.1: $n=m=k$. Then $F(\mathbf{P})=Minimax(\mathbf{P})=\{a,b,c\}$. If $\mathbf{P}'$ is obtained by one voter moving $b$ from last to first place, then in $\mathbf{P}'$, the margin of $b$ over $c$ is $n+2$, that of $c$ over $a$ is $n$,  and that of $a$ over $b$ is $n-2$. It follows that the margin graph of $\mathbf{P}'$ does not satisfy condition (i) under any bijection from $\{a,b,c\}$ to $\{x,y,z\}$, so $F(\mathbf{P}')=Minimax(\mathbf{P}')=\{b\}$.  On the other hand, if $\mathbf{P}'$ is obtained from $\mathbf{P}$ by adding one voter who ranks $b$ first, then either this voter's ranking is $bac$ or $bca$. If it is $bac$, then condition (i) still holds for $\mathbf{P}'$ with $a\mapsto x$, $b\mapsto y$, and $c\mapsto z$, but condition (ii) does not hold for $\mathbf{P}'$, due to the presence of $bac$, so $F(\mathbf{P}')=Minimax(\mathbf{P}')=\{b\}$. If it is $bca$, then in the margin graph of $\mathbf{P}'$, the margin of $b$ over $c$ is equal to the margin of $c$ over $a$.  It follows that the margin graph of $\mathbf{P}'$ does not satisfy condition (i) under any bijection from $\{a,b,c\}$ to $\{x,y,z\}$, so $F(\mathbf{P}')=Minimax(\mathbf{P}')=\{b\}$.  The proofs for $a$ and $c$ are analogous.

Case 2.2.2: $n<m$, $n<k$, and $k-m\leq m-n$. Since $n<m$ and $n<k$, we have $F(\mathbf{P})=Minimax(\mathbf{P})=\{b\}$. Now suppose $\mathbf{P}'$ is obtained by one voter moving $b$ from last to first. Since in $\mathbf{P}'$, the margin of $a$ over $b$ is smaller than that of $b$ over $c$ and that of $c$ over $a$, if $\mathbf{P}'$ satisfies (i), it must be under the bijection $a\mapsto x$, $b\mapsto y$, and $c\mapsto z$. Now since in $\mathbf{P}$, we have $k-m\leq m-n$, i.e., $Margin_{\mathbf{P}}(b,c)- Margin_{\mathbf{P}}(c,a)\leq Margin_{\mathbf{P}}(c,a) - Margin_{\mathbf{P}}(a,b)$, we still have $Margin_{\mathbf{P}'}(b,c)- Margin_{\mathbf{P}'}(c,a)\leq Margin_{\mathbf{P}'}(c,a) - Margin_{\mathbf{P}'}(a,b)$. Hence  (i) does not hold, so $F(\mathbf{P}')=Minimax(\mathbf{P}')=\{b\}$. On the other hand, suppose $\mathbf{P}'$ is obtained from $\mathbf{P}$ by adding one voter who ranks $b$ in first place. Again note that since in $\mathbf{P}'$, the margin of $a$ over $b$ is smaller than that of $b$ over $c$ and that of $c$ over $a$, if $\mathbf{P}'$ satisfies (i), it must be under the bijection $a\mapsto x$, $b\mapsto y$, and $c\mapsto z$. Now the ranking of the new voter in $\mathbf{P}'$ is either $bac$ or $bca$.  First, suppose it is $bac$. If (i) does not hold for $\mathbf{P}'$, then $F(\mathbf{P}')=Minimax(\mathbf{P}')=\{b\}$. If it does, then due to the presence of $bac$ in $\mathbf{P}'$,  (ii) does not hold for $\mathbf{P}'$, so $F(\mathbf{P}')=Minimax(\mathbf{P}')=\{b\}$.  Now suppose the new voter's ranking is $bca$. Then by the assumption of the case, it follows that  $Margin_{\mathbf{P}'}(b,c)-Margin_{\mathbf{P}'}(c,a)\leq Margin_{\mathbf{P}'}(c,a)-Margin_{\mathbf{P}'}(a,b)$. Hence (i) does not hold, so $F(\mathbf{P}')=Minimax(\mathbf{P}')=\{b\}$.\end{proof}

Together Lemmas \ref{ThreeAxLem}-\ref{BPLem} establish Proposition \ref{Independence}.

\section{Beyond three alternatives}\label{BeyondThree}

The obvious next step is to go beyond three alternatives. One might argue that some of the axioms in Theorem~\ref{MainThm} ought to be qualified when applied to profiles with four or more alternatives.\footnote{See, e.g., the discussion of \textit{partial} immunity to spoilers in \citealt[\S~4]{HP2023}.} However, it is certainly of interest to consider the axioms as stated. Table~\ref{AxFig} below shows satisfaction/violation of the key axioms in \textit{arbitrary profiles} by five voting methods (from \citealt{Schulze2011}, \citealt{Kemeny1959}, \citealt{Tideman1987}, and \citealt{HP2023,HP2023stable}, respectively) that refine Minimax in three-alternative~profiles.\footnote{For examples of violations, see  \citealt{Perez2001}, \citealt{Schulze2022}, and \href{https://github.com/wesholliday/may-minimax}{https://github.com/wesholliday/may-minimax}. For proofs of satisfaction for Split Cycle, see \citealt{HP2023}. To see that Kemeny satisfies weak positive responsiveness, we give the argument for Kemeny defined on profiles of linear orders. The argument also works for generalizations of Kemeny to non-linear strict weak orders (\citealt{Fagin2006}, \citealt{Zwicker2018}), but we omit the details of these generalizations. Given a profile $\mathbf{P}$, define the \textit{Kemeny score} of a linear order $L$ of $X(\mathbf{P})$ as the sum over all $(x,y)\in L$ of the number of voters who rank $y$ above $x$ (see \citealt[p.~87]{Fischer2016}). Then $Kemeny(\mathbf{P})$ is the set of alternatives from $X(\mathbf{P})$ that are maximal in some linear order on $X(\mathbf{P})$ with minimal Kemeny score. Now suppose $a\in Kemeny(\mathbf{P})$ and $\mathbf{P}'$ is obtained from $\mathbf{P}$ by one voter moving $a$ from uniquely last place to uniquely first place. Given $a\in Kemeny(\mathbf{P})$, $a$ is maximal in some linear order $L$ of $X(\mathbf{P})$ with minimal Kemeny score, so for every $b\in X(\mathbf{P})\setminus\{a\}$, we have $(a,b)\in L$. For each such $b$, the number of voters who rank $b$ above $a$ decreases by 1 from $\mathbf{P}$ to $\mathbf{P}'$, so each term involving $a$ in the summation for the Kemeny score of $L$ decreases by 1. By contrast, if $L'$ does not have $a$ in first place, then the Kemeny score of $L'$ cannot decrease as much from $\mathbf{P}$ to $\mathbf{P}'$. Hence the orders with minimal Kemeny score in $\mathbf{P}'$ all have $a$ in first place, so $Kemeny(\mathbf{P}')=\{a\}$, in line with weak positive responsiveness.} Note, by way of contrast, that all \textit{proper scoring rules} (see \citealt[Def.~2.9]{Zwicker2016}) satisfy weak positive responsiveness and positive involvement but violate (near) immunity to spoilers already for three alternatives.

\begin{table}[h]
\begin{center}
\begin{tabular}{c|c|c|c|c|c|c|c}
 & Beat  & Kemeny & Minimax & Ranked  & Split  & Stable   \\
 &         Path        &  &               &               Pairs          & Cycle                 & Voting  \\
 \hline
 weak positive responsiveness & $-$ & $\checkmark$ & $\checkmark$ & $-$ & $-$ & $-$ \\
 positive involvement & $-$ & $-$ &  $\checkmark$ & $-$ & $\checkmark$ & $-$ \\
(near) immunity to spoilers & $-$ & $-$ & $\checkmark$ & $-$ & $\checkmark$ & $-$
\end{tabular}
\end{center}
\caption{satisfaction ($\checkmark$) or violation ($-$) of axioms for arbitrary profiles. All violations are known to happen already for 4 alternatives except for Beat Path's violation of immunity to spoilers, which requires 5 alternatives, and Stable Voting's violation of weak positive responsiveness, for which the smallest known violation has 8 alternatives. Methods with $-$ in the third row violate even near immunity to spoilers while methods with $\checkmark$ in the third row satisfy immunity to spoilers.}\label{AxFig}
\end{table}

\begin{question}\label{BeyondThreeQ} Do the axioms of Theorem \ref{MainThm} characterize a class of refinements of Minimax for more than three alternatives?
\end{question}

\noindent Of course, Minimax's satisfaction of the axioms in Theorem \ref{MainThm} in arbitrary profiles must be weighed against its violation---in profiles with four or more alternatives---of other axioms, such as the \textit{Condorcet loser criterion} (never select an alternative that has negative margins versus every other alternative), satisfied by the other voting methods in Table~\ref{AxFig}.\footnote{The other methods in Table \ref{AxFig} also satisfy the stronger \textit{Smith criterion}. The \textit{Smith set} of a profile $\mathbf{P}$, $Smith(\mathbf{P})$, is the smallest set of alternatives such that every alternative in the set has a positive margin over every alternative outside the set. The Smith criterion is the requirement on voting methods $F$ that $F(\mathbf{P})\subseteq Smith(\mathbf{P})$ for all $\mathbf{P}\in\mathrm{dom}(F)$. We note that the Smith-Minimax method, which chooses the Minimax winners from $\mathbf{P}_{\mid Smith(\mathbf{P})}$, violates positive involvement already for four alternatives (this follows from the impossibility theorem of \citealt{Holliday2024}).} 

We already have some hints concerning Question \ref{BeyondThreeQ}. Say that a voting method  $F$ satisfies \textit{ordinal margin invariance} if $F(\mathbf{P})=F(\mathbf{P}')$ for any profiles $\mathbf{P},\mathbf{P}'$ that have the same \textit{ordinal margin graph}, which is obtained from their margin graphs by replacing the edge weights by the strict weak order of their edges by edge weight. Ordinal margin invariance implies several of our previous axioms, namely anonymity, homogeneity, and block preservation. Moreover, all of the voting methods in Table \ref{AxFig} except for Kemeny satisfy ordinal margin invariance. Now consider the following impossibility theorem from \citealt{Holliday2024}: already for four alternatives, there is no voting method satisfying positive involvement, the Condorcet winner and loser criteria, ordinal margin invariance, and resolvability. Here resolvability can be replaced by the condition that the voting method selects a unique winner in any profile that is \textit{uniquely weighted}, meaning that for all alternatives $a,b,c,d$ such that $a\neq b$, $c\neq d$,  $(a,b)\neq (c,d)$,  we have $\mathrm{Margin}_\mathbf{P}(a,b)\neq \mathrm{Margin}_\mathbf{P}(c,d)$ (see \citealt[Lemma 3]{Holliday2024}). Then we can relate the stated impossibility theorem to Question \ref{BeyondThreeQ} as follows.

\begin{proposition} Let $F$ be a voting method on the domain of profiles with up to four alternatives satisfying neutrality, weak positive responsiveness, positive involvement, near immunity to spoilers, and ordinal margin invariance. Then $F$ violates the Condorcet loser criterion.
\end{proposition}

\begin{proof} Suppose $F$ is such a method. First we show that ordinal margin invariance and weak positive responsiveness imply that $F$ selects a unique winner in any uniquely-weighted profile. For suppose there is a uniquely-weighted profile $\mathbf{P}$ with $|F(\mathbf{P})|>1$. Since $\mathbf{P}$ and $5\mathbf{P}$ have the same ordinal margin graphs, we have $F(\mathbf{P})=F(5\mathbf{P})$ by ordinal margin invariance, so $|F(5\mathbf{P})|>1$. Let $\mathbf{P}'$ be obtained from $5\mathbf{P}$ by adding a block of all linear orders, which does not change the margin graph, so $|F(\mathbf{P}')|>1$ by ordinal margin invariance. Finally, let $\mathbf{P}''$ be obtained from $\mathbf{P}'$ by one voter moving some alternative in $F(\mathbf{P}')$ from last place to first place (such a voter exists since we added a block of all linear orders). Then by weak positive responsiveness, $|F(\mathbf{P}'')|=1$. But $\mathbf{P}''$ and $\mathbf{P}'$ have the same ordinal margin graph, since $\mathbf{P}$ was uniquely-weighted and our multiplication of $\mathbf{P}$ by $5$ to obtain $\mathbf{P}'$ made the gaps between margins too big for any one voter to change the ordinal margin graph. Thus, $F(\mathbf{P}'')=F(\mathbf{P}')$, contradicting the facts that $|F(\mathbf{P}')|>1$ and  $|F(\mathbf{P}'')|=1$.

Thus, $F$ selects a unique winner in any uniquely-weighted profile. Moreover, the axioms of the current proposition imply the axioms in Theorem \ref{MainThm}, so by that theorem, $F$ satisfies the Condorcet winner criterion for profiles up to three alternatives. To prove that $F$ satisfies the Condorcet winner criterion for profiles of four alternatives, we reason exactly as in the proof of the first paragraph of Theorem \ref{MainThm} except that we invoke Theorem \ref{MainThm} itself to prove $F(\mathbf{P}''_{-b})=\{c\}$. Thus, $F$ satisfies the Condorcet winner criterion. But then by the main result of \citealt{Holliday2024}, $F$ violates the Condorcet loser criterion.\end{proof}

We suspect that replacing homogeneity and block preservation in Theorem \ref{MainThm} by ordinal margin invariance, which implies both axioms, suffices for axiomatically characterizing the refinements of Minimax for four alternatives. However, ordinal margin invariance is arguably less normatively compelling than the other axioms, so ideally we would have an axiomatic characterization without this axiom.

\section{Conclusion}\label{ConclusionSection}

By adding a compelling set of axioms to those of May's Theorem, Theorem \ref{MainThm} determines how to choose winners in profiles with three alternatives: use Minimax---or some refinement thereof when there are tied margins. Of course, we have not shown that these axioms are the normatively best axioms to augment May's or that Minimax is the normatively best method of voting on three alternatives. However, we believe that Theorem \ref{MainThm} and axiomatic characterizations of other voting methods for three alternatives greatly facilitate the normative comparison of costs and benefits of methods of voting beyond two alternatives. Given recent proposals to implement top-three general elections using Minimax in the United States (\citealt{Foley2024}), such comparisons have not only theoretical but also political importance.

\subsection*{Acknowledgements}
For helpful comments, we thank Yifeng Ding, Mikayla Kelley, David McCune, Milan Moss\'{e},  Zoi Terzopoulou, Bill Zwicker, the two reviewers and editor for \textit{Social Choice and Welfare}, and the audience at the AMMCS 2023 special session on Decisions and Fairness, where this work was first presented.

\subsection*{Compliance with Ethical Standards}
The authors have no conflicts of interest.

\appendix

\section{Minimax ties and refinements of Minimax}\label{Appendix}

In this Appendix, we consider  how we can \textit{properly} refine Minimax on the domain of three-alternative profiles while still satisfying the axioms of Theorem \ref{MainThm} on this domain. Let us say that the \textit{marginal Borda score} of an alternative $a$ in a profile $\mathbf{P}$ is 
\[\underset{b\in X(\mathbf{P})}{\sum} \mathrm{Margin}_\mathbf{P}(a,b).\]
In profiles of linear ballots, the alternatives with the greatest marginal Borda score are also the alternatives with the greatest Borda score in the usual sense (see \citealt[p.~28]{Zwicker2016}), but marginal Borda scores make sense in any profile of  binary relations. Now define Minimax MB (for marginal Borda) as follows: the Minimax MB winners are the Minimax winners with the greatest marginal Borda scores. It is easy to check that Minimax MB satisfies all the axioms in Theorem \ref{MainThm} on the domain of three-alternative profiles. 

For three alternatives, Minimax MB differs from Minimax only in the following cases, using terminology from the proof of Theorem~\ref{MainThm}: 
\begin{itemize}
\item the ascending case when $n=m$ (Minimax MB selects only $c$, whereas Minimax selects $b$ and $c$);
\item the Condorcet loser case when $n<k$ (again Minimax MB selects only $c$, whereas Minimax selects $b$ and $c$).
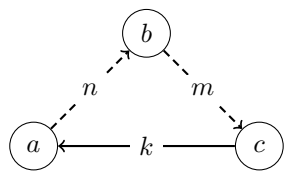
\begin{center}
\begin{tikzpicture}

\node[circle,draw, minimum width=0.25in] at (0,0) (a) {$a$}; 
\node[circle,draw,minimum width=0.25in] at (3,0) (c) {$c$}; 
\node[circle,draw,minimum width=0.25in] at (1.5,1.5) (b) {$b$}; 

\path[->,draw,thick, densely dashed] (b) to node[fill=white] {$m$} (c);
\path[->,draw,thick] (c) to node[fill=white] {$k$} (a);
\path[->,draw,thick, densely dashed] (a) to node[fill=white] {$n$} (b);

\node at (1.5,-.75) {{$0\leq n= m <k$}};
\node at (1.5,-1.5) {{ascending case with $n=m$}};

\end{tikzpicture}\qquad\qquad
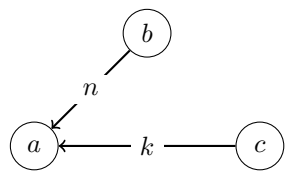
\begin{tikzpicture}

\node[circle,draw, minimum width=0.25in] at (0,0) (a) {$a$}; 
\node[circle,draw,minimum width=0.25in] at (3,0) (c) {$c$}; 
\node[circle,draw,minimum width=0.25in] at (1.5,1.5) (b) {$b$}; 

\path[->,draw,thick] (b) to node[fill=white] {$n$} (a);
\path[->,draw,thick] (c) to node[fill=white] {$k$} (a);

\node at (1.5,-.75) {{$0< n< k$}};
\node at (1.5,-1.5) {{Condorcet loser case with $n<k$}};

\end{tikzpicture}
\end{center}
\end{itemize} 
In the second case, Minimax MB agrees with Stable Voting (\citealt[Prop.~1.1]{HP2023stable}). This case also shows that for three alternatives, Minimax MB satisfies only \textit{near} immunity to spoilers, not immunity to spoilers (since $b$ wins without $a$ in the profile and beats $a$ head-to-head, immunity to spoilers requires $b$ to win in the full profile).\footnote{Moreover, Minimax MB violates even near immunity to spoilers starting at four alternatives. Consider a Condorcet loser case with $n<k$, so Minimax MB selects only $c$; but now add another alternative $d$, whom $b$ and $c$ both beat head-to-head. If $b$'s margin over $d$ is sufficiently large, then although $c$ is the unique winner without $d$ and beats $d$ head-to-head, when $d$ joins, $b$ will be the unique Minimax MB winner, which violates near immunity to spoilers.} To satisfy immunity to spoilers as well as the other axioms of Theorem \ref{MainThm} for three alternatives, one can consider the method that only properly refines Minimax in the ascending case when $n=m$, selecting $c$ in this case. On the other hand, a noteworthy feature of Minimax MB is the following.

\begin{proposition} Minimax MB satisfies positive responsiveness as stated in Remark \ref{PRremark} (\citealt[Def.~7]{Barbera1977}) for any number of alternatives.
\end{proposition}
\begin{proof} Suppose $a$ is among the Minimax MB winners in $\mathbf{P}$, so $a$ is among the Minimax winners with the greatest marginal Borda score. Further suppose $\mathbf{P}'$ is obtained from $\mathbf{P}$ by one voter improving $a$'s position in their ballot, while nothing else changes.\footnote{More precisely, there is some $b\in X(\mathbf{P})$ such that (i) $(a,b)\not\in\mathbf{P}(i)$ changes to $(a,b)\in \mathbf{P}'(i)$ or (ii) $(b,a)\in \mathbf{P}(i)$ changes to $(b,a)\not\in\mathbf{P}'(i)$, while for all alternatives $c\in X(\mathbf{P})\setminus \{a\}$, $(a,c)\in\mathbf{P}(i)$ implies $(a,c)\in\mathbf{P}'(i)$, $(c,a)\not\in \mathbf{P}(i)$ implies $(c,a)\not\in\mathbf{P}'(i)$, and for all  $d\in X(\mathbf{P}) \setminus \{a\}$, $(c,d)\in \mathbf{P}(i)$ if and only if $(c,d)\in\mathbf{P}'(i)$.}  Then each margin of $a$ vs.~another alternative is at least as great in $\mathbf{P}'$ as in $\mathbf{P}$, while at least one such margin is strictly greater in $\mathbf{P}'$; and no other alternative's margin against another alternative increases. Thus, $a$ still has a minimal Minimax score in $\mathbf{P}'$, and its marginal Borda score increases, while no other alternative's marginal Borda score increases. It follows that $a$ is the unique Minimax MB winner in~$\mathbf{P}'$.
\end{proof}

The tradeoff exemplified by Minimax MB vs.~Minimax concerning satisfying positive responsiveness (as opposed to weak positive responsiveness) vs.~immunity to spoilers (as opposed to near immunity to spoilers) is inevitable in the following sense.

\begin{proposition}\label{AppendixProp} There is no voting method that satisfies anonymity, neutrality, positive responsiveness as in Remark \ref{PRremark}, and immunity to spoilers for \textit{all} three-alternative profiles. However, Minimax MB, which satisfies anonymity, neutrality, and positive responsiveness, satisfies immunity to spoilers for all three-alternative profiles with no zero margins between distinct alternatives.\end{proposition}

\begin{proof} For the first claim, suppose for contradiction that there is such a method $F$. Consider the eight-voter profile $\mathbf{P}$ in which each of the six linear orders on $\{a,b,c\}$ is submitted by one voter, while the other two voters, say $i$ and $j$, submit $bca$ and $cba$, respectively. Note that  $\mathrm{Margin}_\mathbf{P}(c,b)=0$ and $\mathrm{Margin}_\mathbf{P}(b,a)=\mathrm{Margin}_\mathbf{P}(c,a)=2$. We claim that $a\not\in F(\mathbf{P})$. Toward a contradiction, suppose $a\in F(\mathbf{P})$. Let $\mathbf{P}^{*}$ be obtained from $\mathbf{P}$ by voter $i$ switching from $bca$ to $bac$, while $j$ still submits $cba$. Then by positive responsiveness, $a\in F(\mathbf{P})$ implies $F(\mathbf{P}^*)=\{a\}$. Let $\mathbf{P}^{**}$ be obtained from $\mathbf{P}^*$ by $i$ switching from $bac$ to $abc$ and $j$ switching from $cba$ to $cab$. Then by positive responsiveness, $F(\mathbf{P}^*)=\{a\}$ implies $F(\mathbf{P}^{**})=\{a\}$. On the other hand, since $\mathbf{P}^*$ consists of each of the six linear orders plus an extra copy of each of $bac$ and $cba$, while $\mathbf{P}^{**}$ consists of each of the six linear orders plus an extra copy of each of  $abc$ and $cab$, so that the roles of $a$ and $b$ are exchanged from $\mathbf{P}^*$ to $\mathbf{P}^{**}$, $F(\mathbf{P}^*)=\{a\}$ implies $F(\mathbf{P}^{**})=\{b\}$ by a standard argument using anonymity and neutrality. Thus, we have a contradiction, so $a\not\in F(\mathbf{P})$. Then since $b$ and $c$ play symmetrical roles in $\mathbf{P}$, it follows by standard arguments using anonymity and neutrality again that $F(\mathbf{P})=\{b,c\}$ and $F(\mathbf{P}_{-a})=\{b,c\}$. Now let $\mathbf{P}'$ be obtained from $\mathbf{P}$ by a voter with the ranking $bac$ switching to $bca$. Then by positive responsiveness, $F(\mathbf{P}')=\{c\}$. Since $\mathbf{P}'_{-a}=\mathbf{P}_{-a}$, we have  $b\in F(\mathbf{P}_{-a}')$. Moreover, $\mathrm{Margin}_{\mathbf{P}'}(b,a)=\mathrm{Margin}_{\mathbf{P}}(b,a)=2$, so $F(\mathbf{P}_{\mid \{a,b\}})=\{b\}$ by Theorem \ref{May}. Then since $a\not\in F(\mathbf{P}')$, immunity to spoilers requires $b\in F(\mathbf{P}')$, contradicting $F(\mathbf{P}')=\{c\}$.

For the second claim, we have already observed that Minimax MB satisfies anonymity, neutrality, and positive responsiveness. For immunity to spoilers with respect to three-alternative profiles with no zero margins between distinct alternatives, suppose $\mathbf{P}$ is such a profile, $a$ is among the Minimax MB winners in $\mathbf{P}_{-b}$, and $a$ is the unique winner in $\mathbf{P}_{\mid \{a,b\}}$. Then since $\mathbf{P}$ has no zero margins between distinct alternatives, it follows that $a$ is the Condorcet winner and hence the Minimax MB winner in $\mathbf{P}$.\end{proof}

\bibliographystyle{plainnat}
\bibliography{social}

\end{document}